\documentclass[sigconf,natbib=false,nonacm]{acmart}

\renewcommand\footnotetextcopyrightpermission[1]{}
\AtBeginDocument{%
  }

\AtEndPreamble{
  \theoremstyle{acmdefinition}
  \newtheorem{remark}[theorem]{Remark}
}

\usepackage[T1]{fontenc}
\usepackage[utf8]{inputenc}
\RequirePackage[datamodel=acmdatamodel,
	style=acmnumeric]{biblatex}
\setcopyright{none}

\usepackage{microtype}
\usepackage[ruled,linesnumbered,vlined]{algorithm2e}
\usepackage[capitalize]{cleveref}
\usepackage{environ}
\usepackage{tikz}

\newcommand{\GspanFP}[3]
    {\langle #1 \rangle _{#2}(#3)}
\newcommand{\casesland}{
\tikz[overlay,baseline]{%
    \path [fill=white,draw=black] (0.59ex, 0.49ex) circle [radius=0.7ex];
    \node () at (0.59ex, 0.49ex) {\scalebox{0.6}{$\land$}};
    }%
}

\newcommand{\caseslor}{
\tikz[overlay,baseline]{%
    \path [fill=white,draw=black] (0.59ex, 0.49ex) circle [radius=0.7ex];
    \node () at (0.59ex, 0.49ex) {\scalebox{0.6}{$\lor$}};
    }%
}

\NewEnviron{landcases}{
  \rlap{$\begin{cases}
    \BODY
  \end{cases}$}
  {\casesland}
  \phantom{
    \begin{cases}
      \BODY
    \end{cases}
  }
}

\NewEnviron{lorcases}{
    \rlap{$\begin{cases}
      \BODY
    \end{cases}$}
    {\caseslor}
    \phantom{
      \begin{cases}
        \BODY
      \end{cases}
    }
  }
\renewcommand{\P}{\mathsf{P}}
\newcommand{\FP}{\mathsf{FP}}
\newcommand{\FPC}{\mathsf{FPC}}
\newcommand{\FO}{\mathsf{FO}}
\newcommand{\ar}{\mathrm{ar}}
\newcommand
    \indexGrp[2]
    {\ensuremath{|#1\mathop{:}#2|}}
\newcommand{\Stab}{\mathrm{Stab}}

\newcommand{\STRUC}{\mathrm{STRUC}}
\newcommand{\Aut}{\mathrm{Aut}}
\newcommand{\range}[1]{[#1]}
\newcommand{\Mod}{\mathrm{Mod}}
\newcommand{\Id}{\mathrm{Id}}
\newcommand{\CPT}{\mathsf{CPT}}

\newcommand{\type}{\mathrm{type}}
\newcommand{\CFI}{\mathsf{CFI}}
\newcommand{\rk}{\mathsf{rk}}

\newcommand{\ord}{\mathsf{ord}}
\newcommand{\freeVar}{\mathrm{free}}
\newcommand{\Sym}{\mathrm{Sym}}
\newcommand{\numb}{\mathrm{numb}}
\newcommand{\element}{\mathrm{elem}}
\newcommand{\graph}{\mathrm{graph}}
\newcommand{\func}{\mathrm{func}}
\newcommand{\perm}{\mathrm{perm}}
\newcommand{\oA}{A^\le}
\newcommand{\stroA}{A^<}
\newcommand{\tq}{\ensuremath{\ |\ }}

\newcommand{\im}{\mathrm{im}}

\newcommand{\oPBCG}{\mathrm{PBCG}^<}

\newcommand\Gspan[1]{\langle #1 \rangle}
\newcommand{\ifp}{\mathrm{ifp}}
\newcommand{\sifp}{\mathrm{sifp}}
\newcommand{\Bij}{\mathrm{Bij}}
\newcommand{\sift}{\mathsf{sift}}
\begin{document}
	\title{Subgroup accessibility in Group Order Logic}
	\author{Anatole Dahan}
	\orcid{0000-0002-6157-8397}
	\affiliation{
		\position{Research Associate}
		\department{Department of Computer Science and Technology}
		\institution{University of Cambridge}
		\city{Cambridge}
		\country{UK}}
	\email{ad929@cam.ac.uk}
\begin{abstract}
	We investigate the expressive power of fixed-point logics ($\FP$) and their extensions in defining generating sets for accessible subgroups of definable permutation groups. This operation, computable in polynomial time via the Schreier-Sims algorithm, plays a central role in the group-theoretic approach to Graph Isomorphism and Graph Canonisation. In particular, it underpins polynomial-time canonisation for bounded colour-class graphs--a class for which no natural logic capturing $\P$ is currently known.
	We first show that this operation cannot, in general, be expressed in any logic for $\P$. This limitation arises from the fact that accessible subgroups need not admit symmetric generating sets of polynomial size.
	However, we prove that when the base group admits a definable ordered generating set, the accessible subgroup operation becomes definable in fixed-point logic with the group order operator ($\FP + \ord$). This is achieved by partially simulating the Schreier-Sims algorithm within $\FP + \ord$.
	As a corollary, we show that fixed-point logic with counting ($\FPC$) can also define the operation when the base group is abelian. In particular, $\FPC$ can define the automorphism group of any graph with abelian colours--despite being unable to canonise such graphs.
\end{abstract}


\maketitle
\section{Introduction}

The quest to identify a logic that precisely characterises the class of problems solvable in polynomial time ($\P$) is a central challenge in descriptive complexity theory.
While fixed-point logic ($\FP$) captures $\P$ on ordered structures~\cite{immermanLanguages1983}, no logic is currently known to capture $\P$ on arbitrary finite structures.
Extensions of $\FP$, such as fixed-point logic with counting ($\FPC$) and fixed-point logic with rank ($\FP + \rk$,~\cite{dawarLogicsRankOperators2009}), have significantly expanded the expressive power of fixed-point logic, but they still fall short of capturing all polynomial-time properties, as shown by increasingly refined separation results~\cite{caiOptimalLowerBound1992,gradelRank2019,lichterSeparating2023}.

In response to these limitations, recent work has investigated group-theoretic extensions of $\FP$.
Notably, Group Order Logic ($\FP + \ord$)~\cite{dahanGroup2025} introduces a new operator that computes the order—that is, the cardinality—of a group generated by a definable set of permutations. This operation is polynomial-time computable using the \emph{Schreier–Sims algorithm}~\cite{simsComputational1970,simsComputation1971}, provided the underlying domain is linearly ordered.

The introduction of the $\ord$ operator was partly motivated by its natural appeal as a generalisation of the $\rk$ operator. Just as $\rk$ defines a fundamental linear-algebraic quantity which is computable efficiently using Gaussian elimination, $\ord$ captures a core computational task in the theory of permutation groups.
However, the analogy between $\rk$ and $\ord$ only goes so far: the $\rk$ operator seems to subsume a wide range of tractable linear-algebraic operations, whereas the Schreier–Sims algorithm enables various operations that have yet to be studied within the framework of $\FP + \ord$.
Most prominently, the Schreier–Sims algorithm can compute generating sets for subgroups of small index of definable permutation groups. Such subgroups are called \emph{accessible}. In this article, we investigate the ability of $\FP + \ord$ to define generating sets for such subgroups.

This study aligns with another motivation for the introduction of the $\ord$ operator: the translation of the group-theoretic approach to Graph Isomorphism and Graph Canonisation within an isomorphism-invariant logical framework.
In the computational framework, the Graph Canonisation problem consists of selecting 
(and computing) a unique representative of each isomorphism class of graphs.
A logic $\mathcal L\ge \FPC$ canonises structures in a class $\mathcal C$ if there is a formula which defines on any structure $\mathfrak A\in\mathcal C$ a copy of $\mathfrak A$ over the numerical domain.
 It follows directly from the Immerman-Vardi~\cite{vardiComplexityRelationalQuery1982} theorem 
that if a logic $\mathcal L\ge \FPC$ canonises structures in a class $\mathcal C$, then $\mathcal L$ captures $\P$ on $\mathcal C$~\cite{groheDescriptiveComplexityCanonisation2017}. Yet, no candidate logic for $\P$ is known to canonise bounded colour-class graphs, a class for which the group-theoretic approach has been shown to yield polynomial-time canonisation algorithms~\cite{babai_monte-carlo_1979,babaiCanonical1983}.
These canonisation results fundamentally rely on the Schreier-Sims algorithm enabling the computation of generating sets for accessible subgroups, and reveal a deep connection between the Graph Isomorphism, Graph Canonisation and Graph Automorphism problems. The latter entails the computation of a generating set for the automorphism group of a graph. It is well-known that Graph Isomorphism can be reduced to Graph Automorphism~\cite{babai_monte-carlo_1979,luksIsomorphismGraphsBounded1982}, and historically, 
breakthroughs on the Graph Automorphism problem have systematically transferred to Graph Canonisation algorithms~\cite{babaiCanonical1983}.
This constitutes a compelling case for investigating under which conditions the automorphism group of a graph can be defined in $\FP + \ord$ or even $\FPC$. In light of the way this problem has been solved for restricted classes of graphs~\cite{babai_monte-carlo_1979,luksIsomorphismGraphsBounded1982}, this in turn raises the question: in which circumstances can a logic define a generating set for an accessible subgroup?

\paragraph*{Contributions}

Our contributions are threefold. First, we establish a fundamental limitation of logical approaches to the definition of sets of generators for accessible subgroups by showing that the operation cannot, in general, be expressed in any logic for $\P$. This arises because 
accessible subgroups need not admit symmetric generating sets of polynomial size. This result suggests 
the inherent unsuitability of the representation of groups by small generating sets in an isomorphism-invariant context.

Second, we identify a natural condition under which this limitation can be overcome. To be precise, our core technical result (\cref{thm:ordered_schreier_sims}) states that when the base group admits a definable \emph{ordered} generating set, the accessible subgroup operation is definable in $\FP + \ord$. Our proof technique involves a partial simulation of the Schreier-Sims algorithm within $\FP + \ord$. When the base group is abelian, we additionally show that the $\ord$ operator is not needed and $\FPC$ already suffices (\cref{corol:use_of_ord_op}).

Taken together, these definability results culminate in our main theorem: on a broad class of coloured structures (including all structures with abelian colours), $\FPC$ can define a generating set for the full automorphism group (\cref{thm:pbcg}).
More precisely, whenever each colour class $A_i$ carries a group $\Gamma_i$ with $\Aut(\mathfrak A_i)\le\Gamma_i\le\Sym(A_i)$ and $|\Gamma_i|\le|A|^k$ for some fixed $k$, and $\Phi$ supplies an ordered generating set for each $\Gamma_i$, $\FPC$ can define a generating set for $\Aut(\mathfrak A)$.
This class of structures encompasses all structures with abelian colours (in the sense of~\cite[Definition 6.1]{pakusaLinear2015}). In particular, our result should be contrasted with the fact that such graphs need not even be distinguishable in $\FPC$~\cite{caiOptimalLowerBound1992} (and hence cannot be canonised in $\FPC$), demonstrating the limitations of the usual reduction from Graph Isomorphism to Graph Automorphism in the present (non union-closed) setting.

While we do not assume commutativity of the $\Gamma_i$, our proof relies on the fact that commutativity is \emph{supplied by the ordering}. As shown in \cref{sec:effect_of_ordering}, when a definable ordered generating set is part of the structure, we can reduce our structure to one where each $\Gamma_i$ is abelian. 
A structure presented with canonically ordered generators thus always has an abelian automorphism group, and the $\ord$ operator, although it is what makes \cref{thm:ordered_schreier_sims} available in general, is never used in \cref{sec:col_graph_aut}.

\paragraph*{Outline}
Section~\ref{sec:preliminaries} introduces the necessary preliminaries, including $\FP + \ord$, the Schreier-Sims algorithm and the notion of subgroup accessibility. In Section~\ref{sec:limits}, we establish the limitations of logical approaches to subgroup accessibility, showing that no logic for $\P$ can define the operation in general.
In Section~\ref{sec:subgroup_comp}, we show that accessible subgroups are definable in $\FP + \ord$ under the assumption of an ordered generating set (\cref{thm:ordered_schreier_sims}), and that in the special case of abelian base groups, $\FPC$ already suffices (\cref{corol:use_of_ord_op}).
In Section~\ref{sec:col_graph_aut}, we apply the results of the previous section to polynomially bounded colour-class graphs whose colour-class generators are ordered. We first establish in \cref{sec:effect_of_ordering} that such an ordering forces $\Aut(\mathfrak A)$ to be abelian (\cref{lem:ordered_centralises,corol:opbcg_abelian}), and then use this to show that $\FPC$ defines a generating set for $\Aut(\mathfrak A)$ on the whole class, with no commutativity hypothesis and without the $\ord$ operator (\cref{thm:pbcg}). In \cref{sec:subgroup_comp}, most proofs are deferred to the appendix.

\section{Preliminaries}
\label{sec:preliminaries}
We denote signatures by upper-case Greek letters, structures by Fraktur symbols (e.g., $\mathfrak{A}, \mathfrak{B}, \mathfrak{C}$), and their respective domains by the corresponding Roman symbols (e.g., $A, B, C$). We assume all structures to be finite, and all signatures to be relational.

Tuples of the form $(v_1, \dots, v_l)$ are denoted by $\vec{v}$. For a set $X$, we write $|X|$ to indicate its cardinality, and for a tuple $\vec{x} = (x_1, \dots, x_k)$, we write $|\vec{x}|$ to denote its length $k$. 
Given $n\in\mathbb N$, we denote $\range{n}$ the set $\{1,2,\dots,n\}$.
We identify a function $f : A\to B$ with its \emph{graph} $\{(a,b)\in A\times B\tq f(a) = b\}$, which we nonetheless write $\graph(f)$ whenever we wish to stress that a function is being viewed as a relation. Conversely, if $R\subseteq A\times B$ is the graph of a function $f: A\to B$, we denote $f$ by $\func(R)$.
Given a function $f : Y\to Z$ and $X\subseteq Y$, we denote by $f_{\restriction X}$ the restriction of $f$ to $X$. 

Given a signature $\Sigma$, the \emph{arity} of $R\in\Sigma$ is denoted $\ar(R)$, and $\STRUC[\Sigma]$ denotes the class of all finite $\Sigma$-structures. Whenever we discuss computational complexity, we tacitly assume some fixed standard encoding of finite $\Sigma$-structures as binary strings; as usual, all \emph{reasonable} encodings are polynomial-time interreducible, so our complexity statements are independent of this choice.
Given a logic $\mathcal{L}$, we denote by $\mathcal{L}[\Sigma]$ the set of formulae in $\mathcal{L}$ over $\Sigma$. We write $\freeVar(\varphi)$ for the set of free variables of $\varphi$, and $\Mod(\varphi)$ for the class of models of a sentence $\varphi$. For a formula $\varphi \in \mathcal{L}[\Sigma]$ and a $\Sigma$-structure $\mathfrak{A}$, we denote by $\varphi(\mathfrak{A})$ the set of assignments $v : \freeVar(\varphi) \to A$ such that $(\mathfrak{A}, v) \models \varphi$. By imposing an ordering on the free variables of $\varphi$, we can view $\varphi(\mathfrak{A})$ as a $|\freeVar(\varphi)|$-ary relation over $A$. 
Typically, this ordering is explicitly specified through the definition of the formula. For instance, if a formula $\varphi(x, y, z)$ is defined, we order the components of $\varphi(\mathfrak{A})$, starting with the $x$-component, followed by $y$, and then $z$. The function that maps $\mathfrak A$ to $\varphi(\mathfrak A)$ is the \emph{query} defined by $\varphi$.
Given two logics $\mathcal L$ and $\mathcal L'$, we write $\mathcal L\le \mathcal L'$ if any query definable in $\mathcal L$ is definable in $\mathcal L'$. 
We denote by
\[\varphi[\psi(\vec x) / \chi(\vec y,\vec x)]\]
the formula $\varphi$ where each occurrence of $\psi(\vec x)$ (for any tuple of variables $\vec x$) is substituted by $\chi(\vec y,\vec x)$ with variables suitably renamed to avoid capture.

To keep the horizontal length of formulae reasonable, we occasionally denote large disjunctions of the form $A\lor B\lor C$ and $A\land B\land C$ as
\[\begin{lorcases} A\\ B\\ C\end{lorcases} \quad\text{and}\quad\begin{landcases} A\\ B\\ C\end{landcases}\quad\text{respectively.}\]

We assume the reader is familiar with first-order logic, which we denote $\FO$. 
$\FP$ is the extension of $\FO$ with an operator $\ifp$ enabling the computation of inflationary fixed-points of definable second-order functions (i.e. mapping relations to relations).
Formally, given any formula $\varphi\in\FP$, $R$ a second-order variable and $\vec x$ a tuple of first-order variables with $\ar(R) = |\vec x|$, $\psi := (\ifp_{R,\vec x} \varphi)$  is also in $\FP$ and has free variables $\freeVar(\varphi)\setminus\{R,\vec x\}$. Given any structure $\mathfrak A$ and valuation $v$ of $\freeVar(\psi)$, $\psi(\mathfrak A,v)$ evaluates to the fixed-point of the sequence $X_n$ defined inductively by $X_0 = \emptyset$ and $X_{i + 1} = \varphi(\mathfrak A, v, X_i)\cup X_i$. 
The Immerman-Vardi theorem~\cite{immermanRelationalQueriesComputable1986} states that, over \emph{ordered} structures, $\FP$ captures $\P$: a class of ordered structures is $\FP$-definable if and only if it is decidable in polynomial time.

We now turn to the definition of $\FPC$, which is an extension of $\FP$ enabling counting. Suppose that $\le \notin \Sigma$. Given any $\Sigma$-structure $\mathfrak A$, its \emph{numerical domain}, denoted $\oA$, is the set $\{0,1,\dots,|A|\}$; it has $|A| + 1$ elements, one more than $A$, so that the counting operator below can return the value $|A|$. We write $\stroA$ for its initial segment $\{0,1,\dots,|A| - 1\}$, which has exactly $|A|$ elements and serves as the domain of canonical copies. We denote by $\mathfrak A^+$ the $\Sigma\sqcup\{\le\}$-structure whose domain is the \emph{disjoint union} $A\sqcup \oA$, where for each $R\in\Sigma$, $R^{\mathfrak A^+} = R^{\mathfrak A}$ and where $\le^{\mathfrak A^+}\subseteq \oA\times \oA$ is the natural linear order on $\oA$. We call $\mathfrak A^+$ the \emph{numerical extension of $\mathfrak A$}. $\FPC[\Sigma]$ is obtained by considering formulae of $\FP[\Sigma\sqcup\{\le\}]$, which, for any $\Sigma$-structure $\mathfrak A$, are to be evaluated on $\mathfrak A^+$.  The elements and numerical domains of the structure interact through the additional \emph{counting operator}: given a formula $\varphi$, $t := (\# x. \varphi)$ is a numerical term with free variables $\freeVar(\varphi)\setminus\{x\}$. Given a valuation $v$ of $\freeVar(t)$, $t(\mathfrak A,v)$ evaluates to $|\{ a\in A \tq (\mathfrak A,v,a)\models \varphi\}|$.

The Immerman-Vardi theorem ensures that all $\P$-computable arithmetic functions can be defined by $\FP$ on $\oA$. 
Note that, with this definition, arithmetic functions involving integers larger than $|A|$ require the encoding of those integers as tuples of numerical values. A $k$-tuple $\vec\mu\in(\oA)^k$ is read as the base-$(|A| + 1)$ expansion $\sum_{i = 1}^{k}(|A| + 1)^{k - i}\mu_i$, so that $k$-tuples encode exactly the integers in $\{0,\dots,(|A|+1)^k - 1\}$. Since $(|A| + 1)^k - 1\ge |A|^k$ for all $k\ge 1$, every integer up to $|A|^k$ is encodable as a $k$-tuple of numerical values, which is the only fact we use.
Given a tuple of numerical variables $\vec\mu$ and an integer $m \le |A|^{|\vec\mu|}$, we write $\vec\mu\gets m$ for the assignment mapping $\vec\mu$ to the unique tuple of numerical values which encodes $m$ in $\mathfrak A$.

We follow usual conventions when writing and handling formulae. In particular, we separate \emph{domain variables}, which range over $A$, composed of \emph{domain elements}, and \emph{numerical variables}, which range over $\oA$, the set of \emph{numerical elements}. This distinction determines the \emph{type} of a variable. 
When reasonable, we keep distinct names for domain variables (for instance $x,y,z$) and numerical variables (for instance $i,j,\mu,\nu,\lambda$). 
However, for reasons that will become clear in the following chapters, this strong separation of variable symbols can often come at the cost of readability. 
In particular, it is often convenient to consider tuples containing variables of different types.
Similarly, we usually distinguish between names for variables and their values. 
When this seems to hinder readability, we may break this rule.
The \emph{type} of a tuple $\vec x$ of variables, denoted $\type(\vec x)$, is the unique word $w\in\{\element,\numb\}^*$ whose $i$-th letter $w_i$ is $\element$ if $x_i$ is a domain variable, and $\numb$ otherwise.
We often need to consider the set underlying all potential valuations of a tuple $\vec x$, and therefore denote by $A^{\vec x}$ (or $A^{\type(\vec x)}$) the set $\prod_{i = 1}^{|\vec x|} A^{\type(\vec x)_i}$, where $A^{\element} := A$ and $A^{\numb} := \oA$. For example, if $\vec x = (y, \mu)$ with $y$ a domain variable and $\mu$ a numerical variable, $A^{\vec x} = A \times \oA$. 

We also allow types instead of arity in the definition of signatures. For instance, if a relation symbol $R$ has type $(\numb,\numb,\element)$, an interpretation of $R$ on $A$ is a subset of $\oA\times\oA\times A$. 
Finally, we extend this notation to relations themselves, so that if $X$ is a relation over $A$, $\type(X)$ is the unique type-word such that $X\subseteq A^{\type(X)}$. 

An \emph{isomorphism} between two $\Sigma$-structures $\mathfrak A,\mathfrak B$ is a bijection $f : A\to B$ such that, for any relation $R\in \Sigma$ and any tuple $\vec a\in A^{\type(R)}$,
\[\vec a\in R(\mathfrak A) \iff f^*(\vec a)\in R(\mathfrak B)\text{ 
where }f^*(\vec a)_i := \begin{cases}
    f(a_i)&\text{if }a_i\in A\\
    a_i&\text{if }a_i\in \oA
\end{cases}\]
Given a logic $\mathcal L\ge\FPC$ and a class $\mathcal C$ of $\Sigma$-structures, $\mathcal L$ \emph{canonises} structures in $\mathcal C$ if there are formulae $(\varphi_R)_{R\in\Sigma}$, each with $|\type(R)|$ free numerical variables, such that, for any structure $\mathfrak A\in\mathcal C$, 
\[(\stroA, (\varphi_R(\mathfrak A)_{R\in\Sigma}))\simeq \mathfrak A.\]
In \cite[Definition 3.3.2]{groheDescriptiveComplexityCanonisation2017}, canonisation is defined in a slightly more general fashion. In the context of this article, this difference is immaterial. 

Given a logic $\mathcal L$ and a class of structures $\mathcal C$, $\mathcal L$ is said to \emph{capture} $\P$ \emph{on} $\mathcal C$ if
\begin{itemize}
    \item for any polynomial-time query $Q$ over the signature $\Sigma$, there is a formula $\varphi\in \mathcal L$ such that, for any $\Sigma$-structure $\mathfrak A\in\mathcal C$, $Q(\mathfrak A) = \varphi(\mathfrak A)$; 
    \item for any sentence $\varphi\in\mathcal L$, there is a polynomial-time algorithm that accepts exactly the encodings of structures in $\Mod(\varphi)$; 
    \item moreover, the map sending a sentence $\varphi\in\mathcal L$ to such a polynomial-time algorithm is itself computable.
\end{itemize}
It follows directly from the Immerman-Vardi theorem that if $\mathcal L\ge \FPC$ canonises structures in $\mathcal C$, then $\mathcal L$ captures $\P$ on $\mathcal C$~\cite{groheDescriptiveComplexityCanonisation2017}. 

Throughout, we use one property of logics repeatedly and in an essential way: \emph{isomorphism-invariance}. A logic $\mathcal L$ is isomorphism-invariant if, whenever $f : \mathfrak A\to\mathfrak B$ is an isomorphism, $f$ maps $\varphi(\mathfrak A)$ onto $\varphi(\mathfrak B)$ for every $\varphi\in\mathcal L[\Sigma]$. Taking $\mathfrak B := \mathfrak A$, this says that the relation defined by any formula on $\mathfrak A$ is preserved by every automorphism of $\mathfrak A$. All logics considered in this article --- and, by definition, every logic in the sense of Gurevich~\cite{gurevich1988logic} --- are isomorphism-invariant. The negative results of \cref{sec:limits} use \emph{nothing else}: they apply to $\FPC$, $\FP + \rk$, $\FP + \ord$, $\CPT$ and infinitary logics alike, and in particular do not depend on any assumption about the complexity of formula evaluation.

\subsection*{Group Order Logic and the Schreier-Sims algorithm}
\label{sec:schreier_sims}
A \emph{permutation} of $A$ is a bijection from $A$ to $A$. A \emph{permutation group} over $A$ is a non-empty set of permutations of $A$ closed under composition and taking inverses. The group of all permutations of $A$ is called the \emph{symmetric group} on $A$, and denoted by $\Sym(A)$. Given $G\le\Sym(A)$, a subset $H\subseteq G$ is a \emph{subgroup} of $G$, denoted $H\le G$, if $H$ is itself a permutation group over $A$. Given $S\subseteq \Sym(A)$, the smallest subgroup of $\Sym(A)$ containing $S$ is denoted $\langle S\rangle$, and $S$ is then said to \emph{generate} $\langle S\rangle$.
If $H\le G$, the relation
\[ x\sim_H y := x^{-1}y\in H\]
defines an equivalence relation on $G$, and this equivalence relation induces an equipartition of $G$. Each part of this partition is called a \emph{left coset of $H$ in $G$}, and we denote by $\indexGrp{G}{H}$ the number of such parts. Clearly $\indexGrp{G}{H}= |G|/|H|$. 
Given $G$ a subgroup of $\Sym(A)$, and $a\in A$, the \emph{orbit of $a$ in $G$} is the set $a^G := \{\sigma(a)\tq \sigma\in G\}$. The \emph{stabiliser of $a$ in $G$} is the subgroup $G_a := \{\sigma\in G\tq \sigma(a) = a\}$.

We study Group Order Logic, which was introduced in~\cite{dahanGroup2025}. It is an extension of $\FPC$ with a new operator $\ord$ that computes the order of a group generated by a definable set of permutations.

The following three definitions build up to that operator.
First, the graph of permutation of $A$ is a binary relation on $A$, hence exactly the sort of object a formula with two free variables defines: 
\begin{definition}
    A formula $\varphi(\vec s,\vec t)$ \emph{defines 
    a permutation $\sigma\in \Sym(A^{\vec s})$ on $\mathfrak A$} if $\varphi(\mathfrak A) = \graph(\sigma)$. 
\end{definition}
To speak of a whole \emph{set} of permutations we let some of the variables act as parameters, each valuation of the parameters picking out one permutation:
\begin{definition}
    A formula $\varphi(\vec p,\vec s,\vec t)$ \emph{defines $S\subseteq \Sym(A^{\vec s})$ binding $\vec p$ in $\mathfrak A$} if
\[ S = \{\sigma\in \Sym(A^{\vec s})\tq \exists \vec a\in A^{\vec p},  \graph(\sigma) = \varphi(\mathfrak A,\vec a)\}.\]
In such a case, we denote the group $\langle S\rangle$ as $\GspanFP{\varphi}{\vec p.\vec s.\vec t}{\mathfrak A}$.
\end{definition}
Given as input such a formula $\varphi$, the $\ord$ operator returns the order of the group $\GspanFP{\varphi}{\vec p.\vec s.\vec t}{\mathfrak A}$ --- a number. As that number might be exponentially larger than $|A|$, it is given in binary over tuples of numerical elements. 

Here and below, we write $\varphi(\vec x)$ to mean that the free variables of $\varphi$ are \emph{exactly} those in $\vec x$; we are explicit whenever additional parameters are allowed.

\begin{definition}
    \label[definition]{def:ord_operator}
Given a $\Sigma$-formula $\varphi(\vec p,\vec s,\vec t)$, with $\type(\vec s) = \type(\vec t) = T$, and a $\Sigma$-structure $\mathfrak A$,
$(\ord_{\vec p.\vec s.\vec t} \varphi)$ is a $2|T|$-ary numerical relation that encodes $|\GspanFP{\varphi}{\vec p. \vec s.\vec t}{\mathfrak A}|$, that is, for any $\vec\mu\in (\oA)^{2|T|}$, 
$\vec\mu\in (\ord_{\vec p.\vec s.\vec t} \varphi)^\mathfrak A$ iff the $\left(\sum_{i = 1}^{2|T|} |A|^{2|T|-i}\mu_i\right)$-th bit of the binary decomposition of $|\GspanFP{\varphi}{\vec p. \vec s.\vec t}{\mathfrak A}|$ is a 1.
\end{definition}
\begin{example}
    Consider 
    \[\varphi(p_1,p_2,s,t) := 
    \begin{lorcases}
        s = p_1 \land  t = p_2\\
        s = p_2 \land t = p_1\\
        s = t \land s\ne p_1\land s\ne p_2
    \end{lorcases}\]
    On any structure $\mathfrak A$ and for any $a,b\in A$, $\varphi(\mathfrak A,a,b)$ is precisely the graph of the transposition $(a\ b)$. That is, on $(\mathfrak A,a,b)$, $\varphi$ \emph{defines} $(a\ b)$, and thus $\varphi$ defines $S := \{(a\ b)\tq a,b\in A\}$ on $\mathfrak A$. Since all permutations can be written as a product of transpositions, we have $\langle S\rangle = \Sym(A)$ and $(\ord_{p_1,p_2,s,t})$ evaluates to $|A|!$ on $\mathfrak A$. 
    In the same way,
    \[\psi(p_1,p_2,p_3,s,t) :=
    \begin{lorcases}
    s = p_1 \land t = p_2\\
    s = p_2 \land t = p_3\\
    s = p_3 \land t = p_1\\
    s = t \land s \ne p_1\land s\ne p_2\land s \ne p_3
    \end{lorcases}\]
    defines the set of all 3-cycles binding $p_1,p_2,p_3$, which generates the \emph{alternating group} over $A$, and thus $(\ord_{p_1,p_2,p_3.s.t}\psi)$ evaluates to $|A|!/2$.
\end{example}
It is shown in~\cite{dahanGroup2025} that $\FP + \ord > \FP + \rk$. Moreover, $\FP + \ord$ has polynomial-time model-checking, making it a new candidate for a logic capturing $\P$ on arbitrary finite structures. 
The $\ord$ operator is computable in polynomial time using the Schreier-Sims algorithm~\cite{simsComputational1970,simsComputation1971}. As this algorithm plays an important role in our results, we present it in detail below. 

Let $S$ be a set of permutations over a finite set $A$, and let
\[\langle S\rangle = H_0 \ge H_1 \ge H_2 \ge \dots H_m = 1\]
be a descending chain of subgroups of $\Sym(A)$. A \emph{Strong Generating Set} (SGS) for $S$ relative to $(H_i)$ is a family $(T_i)_{i < m}$ of sets of permutations such that $T_i$ is a \emph{transversal} of $H_{i + 1}$ in $H_i$, i.e., $T_i$ contains exactly one representative of each coset of $H_{i + 1}$ in $H_i$.
A SGS acts as a \emph{factorising scheme} for $\langle S\rangle$, in the sense that any $\sigma\in \langle S\rangle$ admits a \emph{unique} representation as a product of elements of $T_i$ for $i < m$. As such, $|\langle S\rangle| = |T_0|\cdot |T_1|\cdots |T_{m - 1}|$.
Moreover, given $(T_i)$, and provided membership in each $H_{i+1}$ can itself be tested in polynomial time for elements of $H_i$, this unique representation can be computed in polynomial time via the \emph{sifting procedure}, depicted in \cref{alg:sift}, which is an essential subroutine of the Schreier-Sims algorithm.
We follow the presentation in which a SGS records complete transversals $T_i$. While this is not the most efficient way to implement the algorithm, it suffices to place the problems at hand within $\P$. 

\begin{algorithm}
    \caption{Sifting Procedure}
    \label{alg:sift}
    \DontPrintSemicolon
    \KwIn{$(T_i)_{i< m}$, a SGS for $\langle S\rangle$ relative to $(H_i)_{i=0}^m$ and $\tau\in\Sym(A)$}
    \KwResult{$\mathbf{true}$ if $\tau\in\langle S\rangle$; otherwise a pair $(g,\lambda)$ recording the level $\lambda$ at which sifting failed, together with the residue $g$}
    $g := \tau$\;
    \For{$\lambda = 0$ \KwTo $m - 1$}{
        \If{$\exists! \sigma\in T_\lambda$  s.t. $(\sigma^{-1}\cdot g)\in H_{\lambda + 1}$}{
            $g \gets \sigma^{-1}\cdot g$\;        
        }
        \Else{
            \Return{$(g,\lambda)$}\;
        }
        }
        \lIf{$g = \Id$}{\Return{true}}
        \lElse{\Return{$(g,m)$}}
\end{algorithm}
Although SGSs can be constructed over any chain of subgroups, we are particularly interested in cases where the sifting procedure runs in polynomial time (w.r.t. $|A|$), a condition captured by the following definition: 
\begin{definition}
    \label[definition]{def:adequacy}
    A subgroup chain $H_0 \ge H_1 \ge \dots \ge H_m$ is $k$-\emph{adequate} if the following conditions hold:
\begin{itemize}
    \item Given $\sigma\in H_i$, whether $\sigma\in H_{i + 1}$ can be decided in time $O(|A|^k)$.
    \item $m < |A|^k$.
    \item For all $i < m$, the index $|H_i : H_{i + 1}| := |H_i| / |H_{i + 1}|$ is bounded by $|A|^k$.
\end{itemize}
\end{definition}

The second part of the Schreier-Sims algorithm is the \emph{construction procedure}, which computes a SGS for a given set of permutations $S$ and a subgroup chain $(H_i)_{i < m}$ from $\langle S\rangle$ to $1$, the trivial group.
The procedure is depicted in \cref{alg:sgs_cons}.
\begin{algorithm}
    \caption{Construction procedure}
    \label{alg:sgs_cons}
    \DontPrintSemicolon
    \KwIn{$S\subseteq \Sym(A)$ and $(H_i)_{i\le m}$, a subgroup chain for $\langle S\rangle$.}
    \KwOut{$(T_i)_{i < m}$, a SGS for $(H_i)$}
    $(T_i) := (\{\Id\},\{\Id\},\dots,\{\Id\})$\;
    $\ell := $ the elements of $S$, as a list\;
    \While{$\ell \ne\emptyset$}{
      $\tau :=  \ell.\mathrm{pop}()$\;
      \If{$\sift((T_i),\tau) = (g,i)$}{
      $T_i := T_i\cup\{g\}$\;
      $\ell.\mathrm{add}(
        \bigcup_{j \le i} gT_j \cup 
        \bigcup_{j \ge i} T_jg
      )$\;
      }
    }
    \Return{$(T_i)_{i < m}$}\;
    \vspace{1.635em}
\end{algorithm}

Here $\ell$ is a list: a mutable finite sequence supporting $\mathrm{pop}()$, which removes and returns its first element, and $\mathrm{add}(\cdot)$, which appends a set of permutations.
The elements of $S$ are therefore processed \emph{one after another}, and \cref{alg:sgs_cons} must fix some order in which to do so. Correctness does not depend on which order is chosen --- any order yields a SGS --- and in the classical setting the order is simply inherited from the encoding of the input. It is precisely because no such linear order comes for free in an isomorphism-invariant setting that the choice becomes an obstacle in \cref{sec:subgroup_comp}. 
The operation on line 7 is of particular importance in \cref{sec:subgroup_comp}.
It acts as a saturation procedure, ensuring that the resulting $(T_i)$ is indeed a SGS.
In particular, it ensures that for any two permutations $\sigma,\tau\in\bigcup T_i$, the sifting procedure accepts the products $\sigma\tau$ and $\tau\sigma$.
We refer the reader to~\cite{simsComputational1970,simsComputation1971} for a detailed discussion on the correctness of those two procedures. The runtime of this algorithm was studied by Furst, Hopcroft and Luks~\cite{furstPolynomialtimeAlgorithmsPermutation1980}, who showed that when the subgroup chain is $k$-adequate, those two procedures run in time $O(|A|^{f(k)})$ for some fixed function $f$ depending only on $k$.

Note that when a linear order on $A$ is provided, one can take as $(H_i)$ the \emph{chain of stabiliser subgroups of $\langle S\rangle$}: 
\begin{equation}\label{eqn:stab_adequate_chain}
    H_i := \{\sigma\in \langle S\rangle \tq\forall j \le i, \sigma(a_{j}) = a_{j}\}
\end{equation}
where $a_j$ is the $j$-th element of $A$. This yields the polynomial-time computability of the group order and group membership problems.

The Schreier-Sims algorithm enables a third operation of particular importance to our work, which relies on the following definition: 
\begin{definition}\label[definition]{def:accessibility}
A subgroup $K$ is \emph{$k$-accessible} from $G$ if there is a $k$-adequate subgroup chain $G = H_0 \ge H_1 \ge \dots \ge H_m = K$. 
\end{definition}
If $K$ is a $k$-accessible subgroup of $G$, the Schreier-Sims algorithm enables the computation of a generating set for $K$, given as input a generating set $S$ for $G$.
Indeed, suppose we are given a $k$-adequate subgroup chain $\langle S\rangle = K_0 \ge K_1 \ge \dots \ge K_l = K$ witnessing the accessibility of $K$. Then, if $(H_i)_{i \le m}$ with $H_m = 1$ is a $k$-adequate subgroup chain for $\langle S\rangle$, the following chain is also $k$-adequate: 
\[ \langle S\rangle \ge K_1\ge \dots\ge K_l = K  \ge H_1\cap K \ge H_2\cap K \ge \dots \ge H_m\cap K = 1.\]
Applying the Construction procedure to $S$ and this new chain yields a SGS $(T_i)_{i < m + 1}$ for $\langle S\rangle$, and $\langle \bigcup_{i = l}^{m} T_i \rangle = K_l$. Formally:
\begin{lemma}[\cite{furstPolynomialtimeAlgorithmsPermutation1980}]
    \label[lemma]{lem:accessible_schreier_sims}
    Let $\mathcal K$ be a class of structures, and $\mathbf{G,H}$ two functions mapping any $\mathfrak A\in\mathcal K$ to two groups $\mathbf H(\mathfrak A)\le \mathbf G(\mathfrak A)\le \Sym(A^d)$ for some fixed $d$.
    Suppose that:
    \begin{itemize}
        \item there is a polynomial-time algorithm which computes a generating set for $\mathbf G(\mathfrak A)$ on input any encoding of $\mathfrak A$; 
        \item for some fixed $k$, for all $\mathfrak A\in\mathcal K$, $\mathbf H(\mathfrak A)$ is $k$-accessible from $\mathbf G(\mathfrak A)$. 
    \end{itemize}
    Then, there is a polynomial-time algorithm which computes a generating set for $\mathbf H(\mathfrak A)$ on input any encoding of $\mathfrak A$.
\end{lemma}
While the $\ord$ operator clearly expresses the order of a group and enables membership testing~\cite[Lemma III.3]{dahanGroup2025}, it is not clear whether generating sets for accessible subgroups can always be defined in $\FP + \ord$. This is the main technical question we address in this article.

Let us first formally define the problem instances we aim to study. Looking at~\cref{def:adequacy,def:accessibility}, notice that part of this definition is computational: we expect membership in $G_{i + 1}$ to be decidable in polynomial-time for elements in $G_i$. In order to express the ability of a logic $\mathcal L$ to define generating sets of accessible subgroups as a closure property (in the flavour of \cref{lem:accessible_schreier_sims}), we restrict our attention to those accessible subgroups for which those membership tests are \emph{definable} in $\mathcal L$.
\begin{definition}
  \label[definition]{def:witnessed_accessibility}
    Let $\Sigma$ be a signature, $\mathcal K\subseteq \STRUC[\Sigma]$, and suppose that $\mathbf{G,H}$ are functions, mapping any structure $\mathfrak A\in \mathcal K$ to groups $\mathbf H(\mathfrak A)\le\mathbf G(\mathfrak A)\le \Sym(A^T)$ for some fixed type $T$.

    A logic $\mathcal L$ \emph{witnesses the $k$-accessibility of $\mathbf H$ from $\mathbf G$ in $\mathcal K$} if there is an $\mathcal L[\Sigma]$ formula $\varphi_\in(\vec \mu,X)$ where $\vec\mu$ is a tuple of numerical variables, and $X$ is a second-order relation variable of type $T^2$, such that for any $\mathfrak A\in\mathcal K$ with $n := |A|$:
    \begin{itemize}
        \item The following defines a chain of subgroups $(\mathbf G_i(\mathfrak A))_{i = 0}^{n^k}$:\begin{itemize}
            \item $\mathbf G_0(\mathfrak A) := \mathbf G(\mathfrak A)$
            \item $\mathbf G_{i + 1}(\mathfrak A) := \{\sigma\in \mathbf G_i(\mathfrak A)\tq (\mathfrak A,i,\graph(\sigma))\models \varphi_\in\}$, where $i$ is interpreted as a $k$-tuple of numerical values (assigned to $\vec\mu$) 
        \end{itemize} 
        \item $\mathbf G_{n^k}(\mathfrak A) = \mathbf H(\mathfrak A)$
        \item For all $i<n^k$, $\indexGrp{\mathbf G_i(\mathfrak A)}{\mathbf G_{i + 1}(\mathfrak A)}\le n^k$
    \end{itemize}
\end{definition}
In \cref{sec:limits}, we exhibit pairs $(\mathbf G,\mathbf H)$ such that $\FPC$ witnesses the accessibility of $\mathbf H$ from $\mathbf G$ and defines generating sets for $\mathbf G$, while \emph{no logic} can define a generating set for $\mathbf H$. In \cref{sec:subgroup_comp}, we show on the other hand that if we additionally assume that $\FP + \ord$ defines an \emph{ordered} generating set for $\mathbf G$, then it defines an (ordered) generating set for $\mathbf H$.
In \cref{sec:col_graph_aut}, we use this property of $\FP + \ord$ to show its ability to define the automorphism groups of a large class of graphs.
\section{The undefinability of generating sets for accessible subgroups}
\label{sec:limits}
We investigate the ability of a logic $\mathcal L$ to define generating sets for which it witnesses accessibility. 
We show in this section that isomorphism-invariance is, in this regard, incompatible with the representation of groups as generating sets of permutations. 

Indeed, it is a direct consequence of isomorphism-invariance that if $\varphi$ is a formula, and $\tau$ an automorphism of $\mathfrak A$, it must hold that $\varphi(\mathfrak A)^\tau = \varphi(\mathfrak A)$
where, if $R$ is a relation on $A$, \[R^\tau := \{(\tau(a_1),\dots,\tau(a_k))\tq (a_1,\dots,a_k)\in R\}\] is the image of $R$ under $\tau$.
If such a relation $R$ is the graph of a permutation $\sigma$, $R^\tau$ is the graph of $\sigma^{\tau^{-1}} = \tau\sigma\tau^{-1}$:
\[R^\tau = \{(\tau(a),\tau\sigma(a))\tq a\in A\} = \{ (a,\tau\sigma\tau^{-1}(a))\tq a\in A\}.\]
This implies that any set of permutations definable in a logic must be closed under conjugation by any automorphism of the structure. We call such a set of permutations \emph{canonical} (w.r.t. $\mathfrak A$).
Formally:
\begin{lemma}
    \label[lemma]{lem:canonical_generating_set}
    Consider a logic $\mathcal L$ and let $\varphi\in\mathcal L[\Sigma]$ be a formula with $k + 2$ free variables. For $\mathfrak A\in \STRUC[\Sigma]$, let
    $S^{\mathfrak A} := \{\perm(\varphi(\mathfrak A,\vec a)) \tq \vec a\in A^k\}$ be the set of permutations defined by $\varphi$ on $\mathfrak A$ parameterised by the first $k$ variables. Then: 
    \begin{itemize}
    \item $|S| < |A|^k$
    \item For any $\tau\in\Aut(\mathfrak A)$,
    $
        S = \{ \tau\rho\tau^{-1}\tq \rho\in S\} =: S^{\tau}
    $
    \end{itemize}
    That is, $S$ is a canonical set of permutations over $\mathfrak A$ of polynomial size. 
\end{lemma}

In \cref{thm:strong_limit}, we exhibit a class of structures on which $\FPC$ witnesses the accessibility of some group $\mathbf H$, while any canonical set of permutations that generates $\mathbf H$ has exponential size.

As a first step, and to illustrate this result, we show that there are some groups for which a generating set can be computed in polynomial time, while no small generating sets of permutations for those groups are canonical. 
\begin{lemma}
    \label[lemma]{lem:symmetric_generators_limit}
    There is a class of graphs $\mathcal K$ and a function $\mathbf H$ mapping $\mathfrak A\in \mathcal K$ to a subgroup of $\Sym(A)$ such that:
    \begin{itemize}
        \item There is a polynomial-time Turing machine that computes a generating set for $\mathbf H(\mathfrak A)$ when given as input any encoding\footnote{for any reasonable encoding scheme.} of $\mathfrak A$, for $\mathfrak A\in\mathcal K$. 
        \item For any $\mathfrak A\in\mathcal K$, any generating set $S$ of $\mathbf H(\mathfrak A)$ that is canonical w.r.t. $\mathfrak A$ has size $\ge (|A|/2)!$.
    \end{itemize}
\end{lemma}
\begin{proof}
Consider $\mathcal K := \{ K_{n,n}, n\in\mathbb N\}$, where $K_{n,n}$ is the complete bipartite graph with two parts of $n$ vertices each (represented in \cref{fig:knn}).

\begin{figure}[h]
    \centering
    \includegraphics[width=0.55\linewidth]{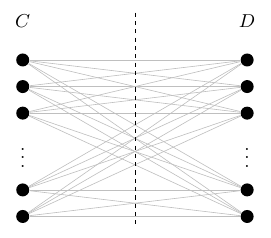}
    \caption{$K_{n,n}$ with its two parts $C$ and $D$.}
    \Description{A schematic complete bipartite graph: two vertical columns of bullet vertices labelled C and D, with light edges between them and vertical dots indicating there are n vertices on each side; a dashed vertical separator marks the bipartition.}
    \label{fig:knn}
\end{figure}

Let $\mathbf H(\mathfrak A) := \Aut(\mathfrak A)$.
It is easy to compute a generating set for $\Aut(\mathfrak A)$ for any $\mathfrak A\in\mathcal K$ in polynomial time: we identify the two parts $C,D$ of $\mathfrak A$, pick any $f\in \Sym(A)$ such that $f(C) = D$ and consider 
\[ S := \{(i\ j)\tq i,j\in C\}\cup \{(i\ j)\tq i,j\in D\}\cup \{f\}.\]
Clearly, this set $S$ can be computed in polynomial time, and $\langle S\rangle = \Aut(\mathfrak A_n)$.
On the other hand, suppose $\langle T\rangle = \Aut(\mathfrak A)$ and $T$ is canonical. Since the side-preserving automorphisms form a subgroup of index $2$, $T$ must contain at least one permutation $\sigma$ with $\sigma(C) = D$; write $u := \sigma_{\restriction C} : C\to D$ and $v := \sigma_{\restriction D} : D\to C$, both bijections.
For $\alpha\in\Sym(C)$, the permutation $\hat\alpha$ acting as $\alpha$ on $C$ and as the identity on $D$ is an automorphism of $\mathfrak A$, and
\[(\hat\alpha\circ\sigma\circ\hat\alpha^{-1})_{\restriction C} = u\circ\alpha^{-1}.\]
As $u$ is a bijection, distinct $\alpha$ yield distinct maps $u\circ\alpha^{-1}$, hence distinct conjugates of $\sigma$. Since $T$ is canonical it contains all of them, so $|T|\ge|\Sym(C)| = n! = (|A|/2)!$.
\end{proof}


The same argument can be adapted to an accessible subgroup, therefore demonstrating that the representation of groups as generating sets is strictly weaker in an isomorphism-invariant context, since this problem can be solved in polynomial time using the Schreier-Sims algorithm (\cref{lem:accessible_schreier_sims}).

\begin{theorem}
    \label[theorem]{thm:strong_limit}
    There is a class of structures $\mathcal K$, two functions $\mathbf G,\mathbf H$ mapping $\mathfrak A\in\mathcal K$ to subgroups of $\Sym(A)$ and a constant $l$ such that:
    \begin{itemize}
        \item There is an $\FPC$ formula $\varphi_{\mathbf G}$ which, on any $\mathfrak A\in\mathcal K$, defines a generating set for $\mathbf G(\mathfrak A)$.
        \item $\FPC$ witnesses the $l$-accessibility of $\mathbf H$ from $\mathbf G$ in $\mathcal K$. 
        \item All canonical generating sets for $\Aut(\mathfrak A)$ have size $\ge 2^{|A|/4}$.
    \end{itemize}
\end{theorem}

\begin{proof}
    We construct, for each $n\in\mathbb N$, a structure $\mathfrak A_n$ such that the class $\mathcal K := \{\mathfrak A_n\tq n\in\mathbb N\}$ satisfies the conditions of the theorem.
    Fix $n$ an integer. For any integer $i$, let $C_i := \{(i,0),(i,1)\}$ and $D_i := \{(i,2),(i,3)\}$. Let $C := \bigcup_{i = 1}^n C_i$ and $D := \bigcup_{i = 1}^n D_i$.
    Consider the $\{\preceq,E\}$-structure $\mathfrak A_n$ such that:
    \begin{itemize}
        \item $A = C\cup D$
        \item $a\preceq^{\mathfrak A_n} b$ iff there are some $i\le j$ such that $a\in C_i\cup D_i$ and $b\in C_j\cup D_j$.
        \item $E^{\mathfrak A_n} = (C\times D)\cup (D\times C)$
    \end{itemize}
Intuitively, $\mathfrak A_n$ is a coloured simple graph with $4n$ vertices such that each colour class induces a subgraph with $4$ vertices, isomorphic to $K_{2,2}$ (equivalently, to a $4$-cycle), namely the complete bipartite graph between $C_i$ and $D_i$. See \cref{fig:colourclasses}. Finally, the whole graph is isomorphic to $K_{2n,2n}$, and each colour-class intersects both of the two parts of this bipartite graph at two distinct vertices. That is, we consider precisely the same structure as in the previous theorem, except a colouring of colour-class size 4 has been added. 

\begin{figure}[h]
    \centering
    \includegraphics[width=0.60\linewidth]{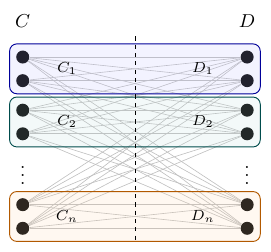}
    \caption{The coloured graph $\mathfrak A_n$: $K_{2n,2n}$ with a colouring into classes of size $4$, each split as $(C_i,D_i)$ across the bipartition.}
    \Description{A complete bipartite graph with two vertical columns of vertices labelled C and D. A dashed vertical line separates the two sides. Three coloured rounded rectangles indicate colour classes, each enclosing two vertices on the left and two on the right, and labelled in reading order as C1, D1, C2, D2, Cn, Dn.}
    \label{fig:colourclasses}
\end{figure}

We now show that $\mathcal K := \{\mathfrak A_n, n\in\mathbb N\}$ has the desired properties.
First, for any $n$, $\mathfrak A_n$ is a $4$-bounded colour-class graph, and as such the automorphism groups of structures in $\mathcal K$ is computable in polynomial-time~\cite{babai_monte-carlo_1979}.

We now show that the $l$-accessibility of $\Aut(\mathfrak A_n)$ is definable in $\FP + \ord$ (uniformly on $n$):
Consider $\mathbf G(\mathfrak A_n):= \prod_{i = 1}^n \Sym(C_i\cup D_i)$. Obviously, $\Aut(\mathfrak A_n)\le \mathbf G(\mathfrak A_n)$.
Moreover, $\mathbf G(\mathfrak A_n)$ is generated by
\[\left\{\left((i,\lambda)\ (i,\mu)\right)\tq i\le n,\ \lambda,\mu\in\{0,1,2,3\}\right\}.\]
We now define $\varphi_\mathbf G$, such that, for any $n$, $\varphi_\mathbf G (\mathfrak A_n)$ defines a generating set of $\mathbf G(\mathfrak A_n)$.
\[\varphi_\mathbf G(p,q,s,t) :=\begin{landcases}
     p\preceq q\\
     q\preceq p\\
     \begin{lorcases}
        s = p\land t = q\\
        s = q\land t = p\\
        s\ne p\land s\ne q\land s = t
     \end{lorcases}
    \end{landcases}
\]
Let us now show that $\Aut(\mathfrak A_n)$ is accessible from $\mathbf G(\mathfrak A_n)$, and that this accessibility is \emph{definable} in $\FPC$. 
Consider $\mathbf G_{0}(\mathfrak A_n) := \mathbf G(\mathfrak A_n)$, and for any $k<n$, let $\mathbf G_{k}(\mathfrak A_n)$ be the subgroup of $\mathbf G(\mathfrak A_n)$ which preserves the edge relation on the union of the first $k$ colour-classes, i.e.\ which stabilises setwise the set of edges having both endpoints in $\bigcup_{i\le k}(C_i\cup D_i)$.
It is easy to see that $\mathbf G_k(\mathfrak A_n)$ is the setwise-stabiliser of the partition $\left\{\bigcup_{i = 1}^k C_i,\bigcup_{i = 1}^k D_i\right\}$ in $\mathbf G(\mathfrak A)$ (for the natural action of $\mathbf G(\mathfrak A_n)$ restricted to those colour-classes).
As such, any two $\sigma,\tau\in\mathbf G_k(\mathfrak A_n)$ belong to the same $\mathbf G_{k + 1}(\mathfrak A_n)$-coset if they act in the same way on both \emph{the partition} $\left\{\bigcup_{i = 1}^k C_i,\bigcup_{i = 1}^k D_i\right\}$, and on \emph{the points in} $C_{k + 1}\cup D_{k+1}$. As there are two possible ways to act on the former, and $|S_4| = 24$ ways to act on the latter, we have $|\mathbf G_k(\mathfrak A_n): \mathbf G_{k + 1}(\mathfrak A_n)| \le 48$, for any $n$.
It only remains to show that a formula $\varphi_{\in}$ that defines conditional membership to $\mathbf G_{k + 1}$ from $\mathbf G_k$ can be constructed in $\FPC$. This can be done as follows: 
\[\varphi_\in(k,X) := \forall x,y\in \bigcup_{i = 1}^{k+1} C_i\cup D_i, \exists x',y', \begin{landcases}
    X(x,x')\\
    X(y,y')\\
    E(x,y)\iff E(x',y')
\end{landcases}
\]
The fact that membership to the first $k$ colour-classes (used here in the quantification of $x$ and $y$) can be defined in $\FPC$ is straightforward. 

Finally, it is left to show that any canonical generating set for $\Aut(\mathfrak A_n)$ has size at least $2^n$. The idea is similar to that of \cref{lem:symmetric_generators_limit}, but the combinatorics of the colour classes makes the bookkeeping slightly more delicate, so we argue by an orbit count rather than by exhibiting the swaps explicitly.

Call $\sigma\in\Aut(\mathfrak A_n)$ a \emph{swap} if $\sigma(C) = D$ (equivalently, $\sigma$ exchanges the two sides of the bipartition rather than preserving them). Since $\mathfrak A_n$ is connected and $\{C,D\}$ is its unique bipartition, every $\sigma\in\Aut(\mathfrak A_n)$ either preserves both $C$ and $D$ or is a swap, and the swaps form a coset of the index-$2$ subgroup of side-preserving automorphisms. In particular, every generating set of $\Aut(\mathfrak A_n)$ contains at least one swap.

Fix a swap $\sigma$. Since $\sigma$ preserves each colour class $C_i\cup D_i$ (being colour-preserving) while exchanging $C_i$ and $D_i$, its restriction $\sigma_i$ to the $4$-element set $C_i\cup D_i$ is a permutation mapping $C_i$ onto $D_i$ and $D_i$ onto $C_i$. There are exactly four such permutations of $C_i\cup D_i$: two are products of two disjoint transpositions, and two are $4$-cycles whose two diagonals are $C_i$ and $D_i$. 

For $i\le n$, let $\theta_i := ((i,0)\ (i,1))$ be the transposition of the two vertices of $C_i$, which is an automorphism of $\mathfrak A_n$ since $\mathfrak A_n$ is complete bipartite between $C$ and $D$ and $\theta_i$ preserves both sides and every colour class. For $J\subseteq\range n$ put $\theta_J := \prod_{i\in J}\theta_i\in\Aut(\mathfrak A_n)$; these $2^n$ automorphisms form an elementary abelian subgroup $\Theta\le\Aut(\mathfrak A_n)$.

We claim that $\Theta$ acts freely\footnote{i.e. $g\cdot x = x$ implies $g = 1$.} on the set of swaps by conjugation, so that every swap has at least $2^n$ conjugates under $\Aut(\mathfrak A_n)$. Indeed, conjugation acts componentwise on colour classes, so it suffices to check that $\sigma_i^{\theta_i} \ne\sigma_i$ for each of the four possible $\sigma_i$. Identifying $C_i\cup D_i$ with $\{0,1,2,3\}$ so that $C_i = \{0,1\}$ and $D_i = \{2,3\}$, we have $\theta_i = (0\ 1)$, and:
\[
  (0\,2)(1\,3)^{\;(0\,1)} = (1\,2)(0\,3),\qquad
  (0\,3)(1\,2)^{\;(0\,1)} = (1\,3)(0\,2),
\]
\[
  (0\,2\,1\,3)^{\;(0\,1)} = (1\,2\,0\,3) = (0\,3\,1\,2),\qquad
  (0\,3\,1\,2)^{\;(0\,1)} = (0\,2\,1\,3).
\]
Hence $\theta_J\sigma\theta_J^{-1} = \sigma$ forces $J = \emptyset$, and the map $J\mapsto\theta_J\sigma\theta_J^{-1}$ is injective.

If $S$ is a canonical generating set for $\Aut(\mathfrak A_n)$, as argued above, $S$ contains a swap $\sigma$. By \cref{lem:canonical_generating_set}, $S$ is closed under conjugation by $\Aut(\mathfrak A_n)$, hence contains the $2^n$ pairwise distinct permutations $\theta_J\sigma\theta_J^{-1}$ for $J\subseteq\range n$. Therefore $|S|\ge 2^n = 2^{|A|/4}$, concluding our proof.
\end{proof}

Having shown that the isomorphism-invariance hinders the representation of accessible subgroups by generating sets, we now turn to the study of a restricted case where this operation can be defined in $\FP + \ord$.

\section{Accessible subgroups of permutation groups with ordered generators}
\label{sec:subgroup_comp}

Looking back at the Schreier-Sims algorithm, we can see that it uses a linear order in two different ways. First, a linear order on the domain of the permutations is needed to consider the chain of stabiliser subgroups described in \cref{eqn:stab_adequate_chain}, which constitutes an adequate chain of subgroups in the sense of \cref{def:adequacy}.
Second, a linear ordering on the generating set itself is implicitly used in \cref{alg:sgs_cons} to choose which elements to insert in the SGS first.\footnote{This is hidden in the use of $\mathrm{pop()}$ on line 4.}
In the classical setting, this ordering on $S$ is simply induced by the ordering of $A$ used to encode the input.

In the context of isomorphism-invariant computation, a linear order on the domain is generally unavailable, preventing the use of point-stabiliser chains. Nevertheless, a group may still admit a definable ordered generating set even when the underlying domain cannot be canonically ordered.
This specific configuration appears in several classes of structures used to separate extensions of $\FPC$ from $\P$, most notably the class of graphs with \emph{abelian colours} introduced by Pakusa~\cite{pakusaLinear2015} (see \cref{sec:col_graph_aut}).
For these structures, $\FPC$ defines an \emph{ordered} generating set for a group $\mathbf G(\mathfrak A)$ from which $\Aut(\mathfrak A)$ is accessible, satisfying the requirement for our algorithm.

In this section, we show that this weaker condition—the presence of a linear order \emph{on the generating set} of the ambient group $\mathbf G$—is sufficient to define a generating set for any $k$-accessible subgroup of $\mathbf G$ in $\FP + \ord$. This result forms the algorithmic basis for the applications to coloured graphs in \cref{sec:col_graph_aut}.
\begin{definition}
\label[definition]{def:ordered_generating_set}
Let $\Sigma$ be a signature, $\mathcal K\subseteq\STRUC[\Sigma]$ and $\mathbf G$ be a function mapping each $\mathfrak A\in\mathcal K$ to a group $\mathbf G(\mathfrak A)\le \Sym(A^T)$ for some fixed type $T$.
A logic $\mathcal L$ \emph{defines an ordered generating set for $\mathbf G$ in $\mathcal K$} if there is a formula $\varphi_{\mathbf G}(\vec\mu,\vec s,\vec t)$ where $\type(\vec s) = \type(\vec t) = T$ and $\vec\mu$ a tuple of \emph{numerical variables} such that, for any $\mathfrak A\in\mathcal K$,
\[\GspanFP{\varphi_{\mathbf G}}{\vec\mu.\vec s.\vec t}{\mathfrak A} = \mathbf G(\mathfrak A).\]
\end{definition}

\begin{theorem}
    \label[theorem]{thm:ordered_schreier_sims}
Let $\Sigma$ be a signature, $\mathcal K\subseteq\STRUC[\Sigma]$ and $\mathbf{G,H}$ be functions, mapping any structure $\mathfrak A\in\mathcal K$ to two groups $\mathbf H(\mathfrak A)\le \mathbf G(\mathfrak A)\le\Sym(A^T)$ for some fixed type $T$.

If $\FP + \ord$ defines an ordered generating set for $\mathbf G$ and witnesses the $k$-accessibility of $\mathbf H$ from $\mathbf G$, then $\FP + \ord$ defines an ordered generating set for $\mathbf H$.
\end{theorem} 

This result relies on a partial simulation of the Schreier-Sims algorithm, which was already given in \cref{alg:sift,alg:sgs_cons}.
More precisely, we build upon the chain of subgroups used to prove \cref{lem:accessible_schreier_sims}: if the accessibility of $H\le \Sym(X)$ is witnessed by the decreasing chain of subgroups $(G_i)_{i = 0}^m$, we can obtain a generating set for $H$ in polynomial time by evaluating \cref{alg:sgs_cons} on the subgroup chain
\[ G_0\ge \dots\ge G_m = H\ge \Stab_H(x_1)\ge\dots\ge \Stab_H(x_1,\dots,x_{|X|}) = 1\]
where $(x_i)_{i = 1}^{|X|}$ is an ordered enumeration of $X$. 
Indeed, in such a case, if \cref{alg:sgs_cons} outputs a SGS $(T_i)_{i = 0}^{m + |X|}$, $\bigcup_{i \ge m} T_i$ constitutes a generating set for $H$.

In our current case however, where $\mathbf{H}(\mathfrak A)\le \Sym(A^T)$, such an ordered enumeration of the domain of the group is not achievable and as such, we cannot expect to simulate those algorithms fully.
 
To overcome this issue, we use the $\ord$ operator to bypass the need for this chain of point-wise stabilisers. That is, our method to prove \cref{thm:ordered_schreier_sims} is to construct iteratively 
a \emph{partially strong} generating set, collecting a coset transversal only for the naturally ordered part of the adequate chain of subgroups
\[\mathbf G_0(\mathfrak A)\ge \dots \ge \mathbf G_{n^k}(\mathfrak A)\]
while we keep an unstructured set of generators for $\mathbf G_{n^k}(\mathfrak A)$.
Here and throughout the section we normalise the chain to have exactly $n^k$ steps, indexed by $i\in\{0,\dots,n^k\}$ so that it has $n^k$ transversals. This is no loss of generality: \cref{def:witnessed_accessibility} bounds the length of the chain by $n^k$, and any shorter chain is padded to that length by repeating its last group, which contributes only trivial transversals.
Therefore, this partially strong generating set structure is made of two parts: a structured array of $m$ rows and $n^k$ columns, defining coset representatives for $\mathbf G_{i + 1}(\mathfrak A)$ in $\mathbf G_i(\mathfrak A)$, and an unstructured list of permutations in $\mathbf H(\mathfrak A)$.
Moreover, the sifting procedure should be adapted according to this new notion of partially strong generating sets.

We structure the proof as follows: in the first subsection, we formalise this notion of partially strong generating sets, and show how it can be represented in extensions of $\FPC$.
Then, we provide an implementation in $\FP + \ord$ of the sifting procedure in this new setting, and conclude with the translation of the construction procedure, and the resulting formula defining a generating set for $\mathbf H(\mathfrak A)$.

\subsection{Partially strong generating sets}
We now show how to adapt the structure of SGS to defer the handling of the stabiliser chain to the $\ord$ operator.
Let us first provide an intuition: a SGS for a chain $G_0\ge\dots\ge G_d = 1$ consists of one transversal per step of the chain. We keep the same information but stop the chain early, at $G_d = \mathbf H(\mathfrak A)$ rather than at the trivial group. The transversals then account for the part of $G_0$ above $\mathbf H(\mathfrak A)$, and what is left over --- the group $\mathbf H(\mathfrak A)$ itself --- is described by a plain generating set. A PSGS is thus a SGS truncated at $\mathbf H(\mathfrak A)$, together with a generating set for the remainder.
\begin{definition}
  \label[definition]{def:psgs}
  Given a $k$-adequate chain of subgroups $(G_i)_{i = 0}^d$ for $G\le \Sym(X)$, a \emph{partially strong generating set (PSGS)} for $(G_i)$ is a pair $(\mathcal R,\mathcal S)$, where:
  \begin{itemize}
    \item $\mathcal R$ assigns to each $i < d$ a transversal $\mathcal R(i)$ of $G_{i + 1}$ in $G_i$, whose elements are moreover indexed as $\mathcal R(i) = \{\mathcal R(i,\lambda)\tq \lambda < |X|^k\}$, the index $\lambda$ ranging over an initial segment; such an indexing exists because $|G_i : G_{i+1}|\le|X|^k$ by $k$-adequacy;
    \item $\mathcal S\subseteq \Sym(X)$ is a generating set for $G_d$, with $|\mathcal S|\le |X|^k$.
  \end{itemize}
  We call $\mathcal R$ the \emph{transversal table} of the PSGS, and $\mathcal S$ the \emph{residual list} of the PSGS.
\end{definition}
That is, a PSGS is exactly the same as a SGS, except we do not require the subgroup chain to end with $G_d = 1$. When this equality fails, a generating set for $G_d$ must be provided instead. 

Throughout this section, let $k$ be a fixed integer and $T$ be a fixed type. Fix a structure $\mathfrak A$ and let $n := |A|$.
In order to prove \cref{thm:ordered_schreier_sims}, we aim to represent a PSGS for the chain:
\[\mathbf G_0(\mathfrak A)\ge \dots \ge \mathbf G_{n^k}(\mathfrak A) = \mathbf H(\mathfrak A).\]
Since, by assumption, both the index of $\mathbf{G}_{i + 1}(\mathfrak A)$ in $\mathbf{G}_{i}(\mathfrak A)$ and the length of this sequence are bounded by $|A|^k$, we can index the permutations in $\mathcal R$ with two tuples of $k$ variables. 
During the definition of the construction procedure, we show that we can bound the size of $\mathcal S$ by $|A|^{2|T|}$.
This yields the following representation of PSGS in $\FP + \ord$:
\begin{definition} 
   Let $(\mathbf K_i)_{i = 0}^{n^k}$ be functions mapping a $\Sigma$-structure $\mathfrak A$ with $|A| = n$ to groups over $A^T$, such that $(\mathbf K_i(\mathfrak A))_{i = 0}^{n^k}$ is $k$-adequate.
  A pair of relations $(R,S)$ over $\mathfrak A$ \emph{represents} a PSGS $(\mathcal R,\mathcal S)$ for $(\mathbf K_i(\mathfrak A))$ if:
\begin{itemize}
  \item $\type(R) = \numb^{2k}\cdot T^2$, and for each $\vec\mu,\vec\nu\in (\oA)^k$, $R^{\mathfrak A}(\vec\mu,\vec \nu) := \{(\vec s,\vec t)\in A^{T\cdot T}\tq (\vec\mu,\vec\nu,\vec s,\vec t)\in R^{\mathfrak A}\}$ is either empty or the graph of a permutation. Intuitively, $\vec\mu$ addresses a row of the transversal table and $\vec\nu$ a position within that row.
  \item For any fixed $\vec\mu$, the set of all permutations $\sigma\in \Sym(A^T)$ such that $\graph(\sigma) = R^{\mathfrak A}(\vec \mu,\vec \nu)$ for some $\vec \nu$ is equal to $\{\mathcal R(i,j),j\le |A|^k\}$, where $i$ is the integer encoded by $\vec \mu$.
  \item $\type(S^{\mathfrak A}) = \numb^{2|T|}\cdot T^2$ and for any $\vec \lambda\in (\oA)^{2|T|}$, $S^{\mathfrak A}(\vec \lambda)$ is either empty or the graph of a permutation.
  \item The set of all permutations $\sigma\in\Sym(A^T)$ such that $\graph(\sigma) = S^{\mathfrak A}(\vec \lambda)$ for some $\vec\lambda$ is equal to $\mathcal S$.
  \item For each $\vec \mu$, the set of $\vec \nu$ such that $R^{\mathfrak A}(\vec\mu,\vec\nu)\ne\emptyset$ form an initial segment of $(\oA)^k$ (for the natural encoding). The same holds for $S^{\mathfrak A}$ w.r.t. $\vec \lambda$. 
\end{itemize}
\end{definition}
This last condition mainly plays a role as an invariant in the iterative definitions to come.
In what follows, we often identify a PSGS $(\mathcal R,\mathcal S)$ with its representation $(R^{\mathfrak A},S^{\mathfrak A})$.

With those definitions in mind, notice that, if we manage to build a representation of a PSGS for $(\mathbf G_i(\mathfrak A))_{i = 0}^{n^k}$, its residual list $\mathcal S$ is by \cref{def:psgs} a generating set for $\mathbf G_{n^k}(\mathfrak A) = \mathbf H(\mathfrak A)$, and it is indexed by tuples of numerical variables, hence \emph{ordered}. This is exactly the aim of \cref{thm:ordered_schreier_sims}.

Remark that we require the rows and columns of both $R^{\mathfrak A}$ and $S^{\mathfrak A}$ to be indexed over the \emph{ordered} domain. 
As a counterpart, we need to show how to order those indexing sets. While this is effectively carried out throughout the following subsections, let us provide a high level intuition for this right away.

By the assumptions of \cref{thm:ordered_schreier_sims}, there is an $\FP + \ord$ formula $\varphi_{\mathbf G}(\vec p,\vec s,\vec t)$ providing an \emph{ordered} enumeration $\mathcal X = (g_1,\dots,g_{n^{|p|}})$ of a generating set for the largest group $\mathbf G(\mathfrak A)$ in the sequence. 
To construct a partially strong generating set for $(\mathbf G_i(\mathfrak A))$, we insert the elements of $\mathcal X$ one by one into an initially trivial PSGS (i.e. where $\mathcal S$ and each list in $\mathcal R$ contain only the identity permutation), so as to obtain, after $\lambda$ such insertions, a PSGS for the chain of subgroups:
\begin{equation}
  \label{eqn:partial_chain_subgroups}
  \Gspan{g_1,\dots, g_\lambda}\cap \mathbf G_0(\mathfrak A)\ge \dots \ge \Gspan{g_1,\dots, g_\lambda}\cap\mathbf{G}_{n^k}(\mathfrak A).
\end{equation}
In \cref{alg:sgs_cons}, a saturation step (given on line 7) was needed to ensure that the SGS did in fact represent a group (in particular, that the membership test defined by the sifting procedure was closed by composition). 
Similarly, to obtain an actual PSGS for the chain of subgroups of \cref{eqn:partial_chain_subgroups}, we need to apply a saturation procedure to the PSGS, that is defined in \cref{lem:sat_sgs}.

Since elements are added iteratively, in an order depending solely on the ordering of $\mathcal X$, we can order cosets of $\mathbf G_{i + 1}(\mathfrak A)$ in $\mathbf G_i(\mathfrak A)$ depending on the order in which elements belonging to those cosets were found. In the same way, we can order the indexing of $S$ depending on the order in which those generators of $\mathbf H(\mathfrak A)$ were found.
Relying on those ordered indexings by induction hypothesis, the saturation procedure can also be conducted in an iterative, canonically ordered fashion.

Having defined the central data structure of our proof, we are now ready to delve into the definition in $\FP + \ord$ of the partial simulation of the Schreier-Sims framework. Let us fix, for the rest of this section, the type $T$ of tuples on which the groups defined by $\mathbf G$ and $\mathbf H$ act.

\subsection{The sifting procedure}
\newcommand\siftRes{\mathsf{siftRes}}
\newcommand\siftLevel{\mathsf{siftLevel}}
Since the construction procedure makes calls to the sifting procedure, it seems sensible to start by defining the latter.
In what follows, suppose that $R$ and $S$ are second-order variables representing the current value of our PSGS. We now define $\Sigma\sqcup\{R\}$-$\FPC$ formulae $\siftRes(\Gamma_\sigma,s,t),\siftLevel(\Gamma_\sigma,\vec \mu)$ which specify the result of the first $n^k$ steps of the sifting procedure on the chain of subgroups represented by $(R^{\mathfrak A},S^{\mathfrak A})$ on input $\sigma$, for any suitable structure $\mathfrak A$.

\begin{lemma}
  \label[lemma]{lem:fpc_sift}
 \emph{(Same assumptions as \cref{thm:ordered_schreier_sims})} There are $\FPC[\Sigma]$ formulae $\siftRes(R,\Gamma_\sigma,\vec s,\vec t)$ and $\siftLevel(R,\Gamma_\sigma,\vec \mu)$ such that, given a structure $\mathfrak A\in\mathcal K$, a transversal table $R^{\mathfrak A}$ over $\mathfrak A$, and a permutation $\sigma\in \Sym(A^T)$:
      \begin{itemize}
        \item $(\mathfrak A,R^{\mathfrak A},\graph(\sigma),\vec a)\models \siftLevel$ iff $\vec a$ encodes the index $i$ of the maximal row of $R^\mathfrak A$ visited during the computation of \cref{alg:sift} on $\sigma$.
        \item If \cref{alg:sift} does not terminate on input $\sigma$ before $\lambda = n^k$, then 
        $\siftRes(\mathfrak A,R^{\mathfrak A},\graph(\sigma))$ is the graph of the value of $g$ before entering the for-loop with $\lambda = n^k$. 
        Otherwise, $\siftRes(\mathfrak A,R^{\mathfrak A},\graph(\sigma))$ is the graph of the permutation which is the last value of $g$ prior to the termination of $\sift(\sigma)$. 
      \end{itemize}\leavevmode
\end{lemma}
In the statement above, when referring to the computation of \cref{alg:sift}, we assume that any choice has been made to extend $R^\mathfrak A$ from a PSGS to a SGS. Since we are only simulating the first $|R| = n^k$ iterations of the sifting procedure, this arbitrary extension does not come into play.

\newcommand\siftEval{\mathrm{siftEval}}
\newcommand\siftBool{\mathrm{sift}_{\mathrm{bool}}}
Now, suppose that $(R,S)$ represent a PSGS $(\mathcal R,\mathcal S)$, and $\sigma\in \Sym(X)$. Assume that we run through \cref{alg:sift}, and let $g_i$ be the value assigned to $g$ when entering the for-loop with $\lambda = i$. One of two outcomes can be reached:
\begin{itemize} 
  \item either for some $i < |\mathcal R|$, there is no coset of $G_{i + 1}$ in $G_i$ which $g_i$ belongs to, which implies that $g_i\not\in G$, and thus $\sigma\not\in G$
  \item or we reach the $|\mathcal R|$-th entry in the loop. Then, $\sigma\in G\iff g_{|\mathcal R|}\in G_{|\mathcal R|}\iff g\in \langle \mathcal S\rangle$.
\end{itemize}
Using $\siftRes$ and $\siftLevel$ together with the $\ord$ operator, it is therefore easy to build a formula $\siftBool(R,S,\sigma)$ which holds on a $n^k$-PSGS representation $(R,S)$ of $G$ iff $\sigma\in G$: 
\[\siftBool(R,S,\sigma) := \siftLevel(\sigma, n^k -1)\land (\siftRes(\sigma,\vec s,\vec t)\in \langle S(\vec s, \vec t)\rangle)_{\vec \lambda,\vec s\vec t}\]
The membership test $(\varphi\in\langle \psi\rangle)$ can be implemented using the $\ord$ operator. Its definition in $\FP + \ord$ can be found in~\cite[Lemma III.3]{dahanGroup2025}.

\subsection{Construction of a PSGS}
To ease the presentation, we divide the construction procedure into smaller parts. We first show how to add a single permutation to a PSGS structure. In a second time, we deal with saturation, and finally, we provide a formula for the whole construction, which iterates the insertion formulae over the generating set of $\mathbf G(\mathfrak A)$ provided by $\varphi_\mathbf G$, and then apply the saturation formulae. Note that the resulting procedure differs slightly from the one given in \cref{alg:sgs_cons}, as we defer all saturation steps to the end of the construction.

\subsubsection{From sifting to insertion}
Before we define the insertion formulae, we direct the reader's attention to the fact that what motivates the existence of a saturation mechanism is that, after a single insertion of a permutation $\sigma$ in a (partial) strong generating set, the resulting structure is not necessarily a (P)SGS anymore.
As such, a definition is necessary to state precisely the post-conditions of our insertion procedure.
\begin{definition}
    \label[definition]{def:psgs_prototype}
    Given a $k$-adequate chain of subgroups $(G_i)_{i = 0}^d$ for $G\le \Sym(X)$, a \emph{PSGS prototype} for $(G_i)_{i = 0}^d$ is a pair $(\mathcal R,\mathcal S)$ such that:
    \begin{itemize}
      \item $\mathcal R$ is a list of length $d$ of lists of permutations. Each of those inner lists contain at most $|X|^k$ permutations. 
      \item $\mathcal S$ is a list of permutations.
      \item The set of permutations contained in $(\mathcal R,\mathcal S)$ generates $G$.
      \item For each $i<d$, the set of permutations in $\mathcal R(i)$ is a subset of $G_i$, and each of those permutation belong to a different coset of $G_{i + 1}$ in $G_i$. $\mathcal S$ is a subset of $G_d$.
    \end{itemize}
\end{definition}
That is, PSGS prototypes differ from PSGS in that we only expect $\mathcal R(i)$ to be a \emph{subset} of a transversal of the cosets of $G_{i + 1}$ in $G_i$. In this context, the additional condition that $\bigcup \mathcal R\cup \mathcal S$ generates $G$ is necessary to maintain a correspondence between PSGS prototypes and the groups they define.
\newcommand
  \ins
  {\ensuremath{{\mathrm{ins}}}}
\newcommand
  \insTable
  {\ensuremath{{\mathrm{ins}_{\mathrm{table}}}}}
\newcommand
  \insList
  {\ensuremath{{\mathrm{ins}_{\mathrm{list}}}}}

We are now ready to introduce the insertion formulae $\ins$, which assume $(R^{\mathfrak A},S^{\mathfrak A})$ to be a PSGS prototype, and evaluate to a larger PSGS prototype:
\begin{lemma}
  \label[lemma]{lem:ins_psgs}
  \emph{(Same assumptions as \cref{thm:ordered_schreier_sims})}
  There is an $\FPC[\Sigma]$ formula $\insTable(R,R_\sigma,\vec\mu,\vec\nu,\vec s,\vec t)$ and an $(\FP + \ord)[\Sigma]$ formula $\insList(R,S,R_\sigma,\vec\lambda,\vec s,\vec t)$ such that, given $\mathfrak A\in\mathcal K$, a subgroup $\mathbf K(\mathfrak A)\le\mathbf G(\mathfrak A)$, a permutation $\sigma\in\mathbf G(\mathfrak A)$ and a PSGS prototype $(R^{\mathfrak A}, S^{\mathfrak A})$ for 
  \[\mathbf K(\mathfrak A)\ge\mathbf K(\mathfrak A)\cap \mathbf G_1(\mathfrak A)\ge\dots\ge \mathbf K(\mathfrak A)\cap \mathbf G_{n^k}(\mathfrak A)\]
  The relations 
  \begin{align*}
  R' &:= \insTable(\mathfrak A, R^{\mathfrak A},\graph(\sigma))\\
  S' &:= \insList(\mathfrak A,R^{\mathfrak A},S^{\mathfrak A},\graph(\sigma))
  \end{align*} form a PSGS prototype for the chain of subgroup
  \[\Gspan{\mathbf K(\mathfrak A),\sigma} \ge 
  \Gspan{\mathbf K(\mathfrak A),\sigma}\cap \mathbf G_1(\mathfrak A)
  \ge \dots
  \ge \Gspan{\mathbf K(\mathfrak A),\sigma}\cap \mathbf G_{n^k}(\mathfrak A)\]
\end{lemma}

In the proof of \cref{lem:ins_psgs} (in \cref{sec:proofs_sc}), we use the ordering of the cosets to order the potential new element to insert within $\mathcal R$.

Iterating insertions over the whole generating set defined by $\varphi_{\mathbf G}$, we obtain:
\newcommand
  \protoTable
  {\ensuremath{{\mathsf{proto}_{\mathrm{table}}}}}
\newcommand
  \protoList
  {\ensuremath{{\mathsf{proto}_{\mathrm{list}}}}}
\begin{corollary}
    \label[corollary]{corol:proto_psgs}
    \emph{(Same assumptions as \cref{thm:ordered_schreier_sims})} There are $\FP + \ord$ formulae $\protoTable(R,S,\vec\mu,\vec \nu,\vec s,\vec t)$ and $\protoList(R,S,\vec\mu,\vec\nu,\vec s,\vec t)$ which define a PSGS prototype for $\mathbf G(\mathfrak A)$ through $(\mathbf G_i(\mathfrak A))_{i = 0}^{n^k}$.
\end{corollary}

\subsubsection{Saturating a PSGS}
We have now defined almost all the parts of the construction procedure.
We now show how to turn a PSGS prototype for some group $\mathbf K(\mathfrak A)$ into a PSGS for $\mathfrak A$.
Remark that the only difference between the two notions is that if $(R,S)$ is a PSGS for some chain of subgroups $(\mathbf K_i(\mathfrak A))$, the $i$-th row of $R$ constitutes a transversal of $\mathbf K_{i + 1}(\mathfrak A)$ in $\mathbf K_i(\mathfrak A)$ (see \cref{def:psgs,def:psgs_prototype}).

By \emph{saturation} we mean the closure operation performed on line~7 of \cref{alg:sgs_cons}.
A prototype is a candidate table whose rows need not yet be transversals: some coset of $\mathbf K_{i+1}$ in $\mathbf K_i$ may be unrepresented, and a product of two entries of the table may fail to sift. Saturating means repeatedly sifting the products $\sigma\tau$ and $\tau\sigma$ of pairs of entries already present and inserting the residues where sifting fails, until no new element is produced. The fixed point of this process is a genuine PSGS: once every such product sifts successfully, each row is a full transversal and the factorisation property of \cref{sec:schreier_sims} holds.

\newcommand
  \sat
  {\ensuremath{{\mathsf{sat}}}}
\newcommand
  \satList
  {\ensuremath{{\mathsf{sat}_{\mathrm{list}}}}}
\newcommand
  \satTable
  {\ensuremath{{\mathsf{sat}_{\mathrm{table}}}}}

\begin{lemma}\label[lemma]{lem:sat_sgs}
  \emph{(Same assumptions as \cref{thm:ordered_schreier_sims})}
  There are $(\FP + \ord)[\Sigma]$ formulae $\satTable(R,S,\vec\mu,\vec s,\vec t)$ and $\satList(R,S,\vec\lambda,\vec s,\vec t)$ such that, given a structure $\mathfrak A$ and a PSGS prototype $(R^{\mathfrak A},S^{\mathfrak A})$ for some group $\mathbf K(\mathfrak A)$,
  $R' := \satTable(\mathfrak A)$ and $S' := \satList(\mathfrak A)$ form a PSGS for the same group.
\end{lemma}

\subsection{Conclusion}
We now conclude our proof that $\mathbf H(\mathfrak A)$ admits an $\FP + \ord$ definable generating set:
\begin{proof}[Proof of \cref{thm:ordered_schreier_sims}]
The $\mathsf{proto}$ formulae from \cref{corol:proto_psgs} already define a PSGS prototype for $\mathbf G(\mathfrak A)$ through $(\mathbf G_i(\mathfrak A))_{i = 0}^{n^k}$.
Applying the $\mathsf{sat}$ formulae then yields an actual PSGS for this chain of subgroup, and the residual list constitutes a generating set for $\mathbf H(\mathfrak A)$:

\[
  \varphi_{\mathbf H}(\vec\lambda,\vec s,\vec t) :=
  \satList(R,S)[R / \protoTable, S / \protoList]\qedhere
\]
\end{proof}

While our proof is now complete, there is a final remark to be made: we have used the $\ord$ operator very scarcely.
Indeed, the $\ord$ operator has only been used in $\insList$ to ensure that we are not adding ``superfluous'' elements to $S$. That is, whenever we add an element to $S$, the resulting set $S'$ generates a larger group, and thus $|\langle S'\rangle|\ge 2\cdot |\langle S\rangle|$. Therefore, the number of elements in the final residual list is necessarily bounded by $\log_2(\Sym(A^T))\le |A|^{|T|}\log(|A|^{|T|})\le |A|^{2|T|}$.

When $\mathbf G_0$ is abelian, we can enforce a polynomial bound on the size of $S$ in another way, without using the $\ord$ operator at all.
As such, we observe that for abelian groups, which play a central role in the definition of the graph classes in \cref{sec:col_graph_aut}, the power of the $\ord$ operator is not required; $\FPC$ alone suffices.
\begin{theorem}
  \label[theorem]{corol:use_of_ord_op}
  Let $\Sigma$ be a signature, $\mathcal K\subseteq \STRUC[\Sigma]$ and suppose that $\mathbf{G,H}$ are two functions, mapping any structure $\mathfrak A\in\mathcal K$ to two groups $\mathbf H(\mathfrak A)\le\mathbf G(\mathfrak A)\le \Sym(A^T)$ for some fixed type $T$.

  Suppose that $\FPC$ defines an ordered generating set for $\mathbf G$ and witnesses the $k$-accessibility of $\mathbf H$ from $\mathbf G$ in $\mathcal K$. Suppose additionally that, for all $\mathfrak A\in\mathcal K$, $\mathbf G(\mathfrak A)$ is abelian.
  Then, there is a formula $\varphi_{\mathbf H}\in \FPC[\Sigma]$ defining an ordered generating set for $\mathbf H$ in $\mathcal K$.
\end{theorem}

\section{Coloured Graph Automorphisms}
\label{sec:col_graph_aut}
  As an application of the results from the previous section, we now exhibit a large class of graphs on which $\FPC$ can define a generating set of the automorphism group.
   

  In~\cite{pakusaLinear2015}, a structure $\mathfrak A$ is said to have \emph{abelian colours} if it is equipped with a total pre-order $\preceq^{\mathfrak A}$ and relation $\Phi^{\mathfrak A}$ such that, for each $i$, 
  $\Phi(\mathfrak A,i)$ \emph{enumerates} an abelian (and ordered) group $\Gamma_i$ acting transitively on $A_i$, the $i$-th colour class of $\mathfrak A$ (w.r.t. $\preceq^{\mathfrak A}$). 
  Let us denote $\mathfrak A_i$ the substructure of $\mathfrak A$ induced by $A_i$.
  Note that here, we require the graph of \emph{any} element of $\Gamma_i$ to appear as $\Phi(\mathfrak A,i,j)$ for some $j$.
  Because $\Gamma_i\le \Sym(A_i)$ is abelian and transitive, it can be shown that $|\Gamma_i| = |A_i|$, which makes this hypothesis reasonable. One can also show that, under those hypotheses, $\Aut(\mathfrak A_i)\le \Gamma_i$. 
  
  We show that, under the conditions that $\Aut(\mathfrak A_i)\le \Gamma_i$ for each $i$ and $\Gamma_i$ is canonically ordered by $\Phi$, all the remaining group-theoretic assumptions on the $\Gamma_i$ may be relaxed while retaining $\FPC$-definability of the automorphism group of $\mathfrak A$.
  Precisely, we introduce in \cref{def:pbcg} the class $\oPBCG_k$ of \emph{Polynomially Bounded Colour-class graphs}: coloured graphs in which each colour class $A_i$ carries a group $\Gamma_i$ with $\Aut(\mathfrak A_i)\le\Gamma_i\le\Sym(A_i)$ and $|\Gamma_i| < |A|^k$, presented by an ordered generating set $\Phi$. \Cref{thm:pbcg} states that $\FPC$ defines a generating set for $\Aut(\mathfrak A)$ on that class.
  Compared to abelian colours, the groups $\Gamma_i$ are thus no longer required to be transitive, nor to be fully enumerated by $\Phi$ rather than merely generated by it, and their size is allowed to grow polynomially (and not only linearly) with $|A|$.
  While commutativity is not required either, we will see in \cref{sec:effect_of_ordering} that it is actually enforced by $\Phi$: an ordered generating set for the ambient group forces $\Aut(\mathfrak A)$ to centralise it, hence to be abelian. It should therefore be noted that abandoning commutativity in the hypotheses does not extend the reach of the method to any structure of non-abelian automorphism group.
  \Cref{thm:pbcg} is to be contrasted with the fact that the class of structures considered contains the $\CFI$-structures from~\cite{lichterSeparating2023}, which separate $\FP + \rk$ from $\P$: in this specific context, a strong distinction seems to exist between the Automorphism problem (which we thus show to be $\FPC$ definable), and the Canonisation problem (which cannot be defined in $\FP + \rk$, by~\cite{lichterSeparating2023}).

  This is due, at least in part, to the fact that the definability of an ordered generating set for the automorphism group of each colour-class is a not union-closed property, which in turn implies that we cannot rely on the usual reduction (see~\cite{babai_monte-carlo_1979,luksIsomorphismGraphsBounded1982}) from the Graph Isomorphism problem to the Graph Automorphism problem.
  Indeed, given $\mathfrak {A,B}$ two such coloured structures, and $i$ the index of a colour-class, it is not clear how to obtain an ordered generating set for $\Aut(\mathfrak A_i \sqcup \mathfrak B_i)$ from the two ordered generating sets given by $\Phi(\mathfrak A,i)$ and $\Phi(\mathfrak B,i)$. In particular, when $\mathfrak A_i\simeq\mathfrak B_i$, defining such an ordered generating set of $\Aut(\mathfrak A_i \sqcup \mathfrak B_i)$ entails to pick one canonical swap between $A_i$ and $B_i$, leading to a situation similar to that of \cref{thm:strong_limit}.
  
  Whether the Isomorphism problem for such structures is definable in $\FP + \ord$ in the general case remains, to our knowledge, open.

  Our proof technique is a direct adaptation of the \emph{Bounded Colour-class Graph} Automorphism algorithm introduced by Babai in~\cite{babai_monte-carlo_1979}, that we review in the next subsection.
  The main obstacle to its application in an isomorphism-invariant framework being that it relies on the Schreier-Sims algorithm enabling the computation of generating sets for accessible subgroups.  
  
  \subsection{From bounded colour-classes to PBCG}

  In order to introduce our generalisation, we review the group-theoretic framework to the automorphism of coloured graphs, introduced in~\cite{babai_monte-carlo_1979}.

  Let $\mathfrak A$ be a graph over a set $A$ with $n$ vertices equipped with a total pre-order $\preceq^{\mathfrak A}$ defining a colouring $c : A \to \range{m}$. We denote by $\mathfrak A_i$ the substructure of $\mathfrak A$ induced by $c^{-1}(i)$, for any $i\in c(A)$. 
  $\Aut_c(\mathfrak A)$ is the group of colour preserving automorphisms of $\mathfrak A$.

  Consider any family $(\Gamma_i)_{i\in\range m}$ of permutation groups with
  \begin{equation}\label{eqn:local_groups}
    \Aut(\mathfrak A_i)\ \le\ \Gamma_i\ \le\ \Sym(A_i)\qquad\text{for every }i\in\range m.
  \end{equation}
  The reader may take $\Gamma_i := \Aut(\mathfrak A_i)$ throughout the present subsection, which is the choice made in~\cite{babai_monte-carlo_1979}; being able to choose another family will be used in the following subsections.
  Let $\mathcal I := \{\{i,j\}\tq i,j\in\range m,\ i\le j\}$, and fix a linear ordering of $\mathcal I$.
  We denote by $s_\lambda$ the $\lambda$-th element of $\mathcal I$ with regard to this ordering.
  For any $s = \{i,j\}\in\mathcal I$, we denote by $\mathfrak A_s$ the substructure of $\mathfrak A$ induced by $c^{-1}(i)\cup c^{-1}(j)$. Note that $\mathcal I$ contains the singletons $\{i\}$, for which $\mathfrak A_{\{i\}} = \mathfrak A_i$.
  We can define the following tower of subgroups:
  \begin{equation}\label{eqn:subgroup_tower_for_colored_graphs}
    G_\lambda :=
      \left\{\sigma \in \prod_{i\in\range{m}} \Gamma_i\ \middle|\
         \forall \kappa <\lambda, \sigma_{\restriction A_{s_{\kappa}}}\in \Aut_c(\mathfrak A_{s_\kappa})
      \right\}
  \end{equation}
   By definition, for  $\lambda,\kappa\le|\mathcal I|$, if $\lambda\le\kappa$  then $G_{\kappa}\le G_{\lambda}$.
  We aim to delimit under which hypothesis this chain of subgroups witnesses the accessibility of $\Aut_c(\mathfrak A)$.
  Notice that $G_{|\mathcal I|} = \Aut_c(\mathfrak A)$, and that membership to $G_{\lambda + 1}$ within $G_{\lambda}$ is clearly $\P$-decidable. Thus, we are left with two conditions to fulfil:
  \begin{enumerate}
    \item\label{item:cond_pbcg_init} A generating set for $G_{0}= \prod_{i\in\range{m}}\Gamma_i$ can be computed in polynomial time.
    \item\label{item:cond_pbcg_rec} There is a constant $\ell$ such that $\forall \lambda<|\mathcal I|, |G_{\lambda}:G_{\lambda + 1}| \le n^\ell$
  \end{enumerate}
  Fulfilling those conditions depends on the restrictions we set on the colour-classes.

  In~\cite{babai_monte-carlo_1979}, a natural class of (classes of) graphs is considered, for which (\ref{item:cond_pbcg_init}) and (\ref{item:cond_pbcg_rec}) hold: \emph{bounded colour-class graphs}.
  A coloured graph $\mathfrak A$ has \emph{$k$-bounded colour-classes} if, for each $i$, $|A_i| \le k$.
  $\mathcal C$ is a class of \emph{bounded colour-class graphs} if there is a constant $k$ such that each graph in $\mathcal C$ has $k$-bounded colour-classes. 
  If $\mathcal C$ is such a class, then, for any $\mathfrak A\in\mathcal C$ and any $i\in\range{m}$, $\Gamma_i := \Aut(\mathfrak A_i)\le \Sym(A_i)$ has at most $k!$ elements, which can be enumerated in polynomial time, and $S := \bigcup_{i\in\range{m}} \Gamma_i$ fulfils (\ref{item:cond_pbcg_init}).
  Now, given $\lambda < |\mathcal I|$, the number of cosets of $G_{\lambda + 1}$ in $G_{\lambda}$ is bounded by the number of ways to permute $A_{s_\lambda}$, itself bounded by $k!^2$, so (\ref{item:cond_pbcg_rec}) is fulfilled.
  
  However, this polynomial-time computation cannot be translated directly to an isomorphism-invariant context. In particular, the class of structures used in the proof of \cref{thm:strong_limit} has 4-bounded colour-classes, and generating sets for their automorphism groups cannot be defined: the enumeration of $\Aut(\mathfrak A_i)$ depicted above cannot be canonically ordered, and thus \cref{thm:ordered_schreier_sims} cannot be applied.

  We now introduce a quite general class of coloured graphs that fulfil (\ref{item:cond_pbcg_init}) and (\ref{item:cond_pbcg_rec}), while providing an ordering similar to that of structures with abelian colours.
  \begin{definition}
    \label[definition]{def:pbcg}
    An instance of $\oPBCG_k$ is a coloured graph $\mathfrak A = (A,E^{\mathfrak A},\preceq^{\mathfrak A},\Phi^{\mathfrak A})$ together with a family of local groups $(\Gamma_i)$ as in~\eqref{eqn:local_groups} with $|\Gamma_i| < n^k$, where $\Phi$ is a relation of type $\numb^{k + 1}\cdot\element^2$ such that for every $i$, $\Phi(\mathfrak A,i)$ defines an ordered generating set for $\Gamma_i$ (in the sense of \cref{def:ordered_generating_set}).
  \end{definition}
  We stress that $\Gamma_i$ is not required to be abelian, nor transitive, nor to consist of automorphisms of $\mathfrak A_i$: it is merely a polynomially bounded group of permutations of the colour class, containing $\Aut(\mathfrak A_i)$, and presented by an ordered generating set. The bound $|\Gamma_i| < n^k$ is what makes the tower~\eqref{eqn:subgroup_tower_for_colored_graphs} an adequate chain, and hence a witness of the accessibility of $\Aut_c(\mathfrak A)$.
  \begin{lemma}
    \label[lemma]{lem:pbcg_poly_index}
    Let $(\Gamma_i)$ be a family of local groups satisfying~\eqref{eqn:local_groups} with $|\Gamma_i| < n^k$ for every $i$, and let $(G_\lambda)$ be the corresponding tower~\eqref{eqn:subgroup_tower_for_colored_graphs}. Then $|\mathcal I|\le n^2$ and, for all $\lambda < |\mathcal I|$, $|G_{\lambda} : G_{\lambda + 1}| < n^{2k}$.
  \end{lemma}
  \begin{proof}
    As $m\le n$, we have $|\mathcal I| = m(m+1)/2\le n^2$.
    For all $\lambda$, let $i,j$ be such that $i \le j$ and $\{i,j\} = s_\lambda$. Then $\pi_{\lambda}(\sigma) := (\sigma_{\restriction A_i},\sigma_{\restriction A_j})$ defines a morphism $\pi_{\lambda} : G_{0}\to \Gamma_i\times \Gamma_j$.
    For $\lambda < |\mathcal I|$ and $\sigma,\tau\in G_{\lambda}$, if $\pi_\lambda(\sigma) = \pi_\lambda(\tau)$, then $\sigma G_{\lambda + 1} = \tau G_{\lambda + 1}$, since we then have
    \[(\sigma^{-1}\tau)_{\restriction A_{s_\lambda}} = 1\in \Aut_c(\mathfrak A_{s_\lambda}).\]
    Therefore
    \[
      |G_{\lambda}:G_{\lambda+1}| \le |\im(\pi_\lambda)|
      \le |\Gamma_i|\cdot |\Gamma_j|
      \le n^k\cdot n^k = n^{2k}.\qedhere
    \]
  \end{proof}
  Remark that we have assumed the generating set for $\Gamma_i$ to be given by a $k + 3$-ary relation.
  In the context where $|\Gamma_i|< n^k$, and $\Phi$ is indexed over the numerical domain, this can be enforced by a simple trick (provided in \cref{sec:proofs_sc}):
  \begin{remark}
    \label[remark]{rem:pcbgkk}
    For $k_1 < k_2$, consider $\oPBCG_{k_1,k_2}$ be the class of co\-loured graphs $\mathfrak A = (A,E^{\mathfrak A},\preceq^{\mathfrak A},\Phi^{\mathfrak A})$ as in \cref{def:pbcg} except that $\Phi$ is a relation of type $\numb^{k_2 + 1} \element^2$ with $\Gamma_i = \GspanFP{\Phi}{\vec p}{\mathfrak A,i}$ and $|\Gamma_i| < n^{k_1}$.
    There is an $\FPC$-interpretation from $\oPBCG_{k_1,k_2}$ to $\oPBCG_{k_1}$.
  \end{remark}

  \subsection{The effect of the ordering assumption}
  \label{sec:effect_of_ordering}

  Before turning to definability, we explain in what sense the ordering of the generating sets enables to assume commutativity.
  Recall from \cref{sec:limits} that if a set of permutations $S$ is definable on $\mathfrak A$, then $\Aut(\mathfrak A)$ acts on that set by conjugation, and $S$ is fixed set-wise by that action of $\Aut(\mathfrak A)$. 
  If we assume that an ordered enumeration of $S$ is definable, the action of $\Aut(\mathfrak A)$ by conjugation must fix $S$ \emph{point-wise}, and hence every automorphism of $\mathfrak A$ must commute with every element of $\langle S\rangle$.
  (Recall from \cref{sec:preliminaries} that an isomorphism $f$ acts on tuples of mixed type by $f^*$, fixing the numerical components; in particular every $\pi\in\Aut(\mathfrak A)$ induces a permutation $\pi^*$ of $A^T$, and $\pi\mapsto\pi^*$ is injective as soon as $T$ contains a domain component.)
  \begin{lemma}
    \label[lemma]{lem:ordered_centralises}
    Let $\mathcal L$ be a logic, $\mathcal K\subseteq \STRUC[\Sigma]$, and let $\varphi(\vec\mu,\vec s,\vec t)\in\mathcal L[\Sigma]$ where $\vec\mu$ is a tuple of \emph{numerical} variables and $\type(\vec s) = \type(\vec t) = T$. Write $\mathbf G(\mathfrak A) := \GspanFP{\varphi}{\vec\mu.\vec s.\vec t}{\mathfrak A}$. Then for every $\mathfrak A\in\mathcal K$ and every $\pi\in\Aut(\mathfrak A)$, the permutation $\pi^*$ commutes with every element of $\mathbf G(\mathfrak A)$.
  \end{lemma}

  \begin{proof}
    Fix $\pi\in\Aut(\mathfrak A)$ and $\vec m\in(\oA)^{|\vec\mu|}$. We have $\varphi(\mathfrak A)^\pi = \varphi(\mathfrak A)$, and $\pi^*$ fixes every numerical element, hence fixes $\vec m$; the identity therefore holds at each numerical index separately, that is $\varphi(\mathfrak A,\vec m)^\pi = \varphi(\mathfrak A,\vec m)$.
    Whenever $\varphi(\mathfrak A,\vec m)$ is the graph of a permutation $\sigma_{\vec m}$, the reasoning opening \cref{sec:limits} yields $\pi^*\sigma_{\vec m}(\pi^*)^{-1} = \sigma_{\vec m}$. These $\sigma_{\vec m}$ generate $\mathbf G(\mathfrak A)$.
  \end{proof}

  Given a group $G$, its \emph{centre} is defined as
  \[ Z(G) := \{g\in G\tq \forall h\in G, g h = h g\}.\]
  In a $\oPBCG_k$ structure the relation $\Phi$ belongs to the signature, so \cref{lem:ordered_centralises} applies, and does so at two levels: to the whole structure $\mathfrak A$, and to each induced substructure $\mathfrak A_i$ separately, since $\Phi(\mathfrak A,i)$ is a relation of $\mathfrak A_i$ defining an ordered generating set for $\Gamma_i$ on it. The two applications give the following.

  \begin{corollary}
    \label[corollary]{corol:opbcg_abelian}
    Let $\mathfrak A\in\oPBCG_k$ with colour classes $A_1,\dots,A_m$, and set $\mathbf G(\mathfrak A) := \prod_{i\in\range m}\Gamma_i$. Then
    \begin{equation}\label{eqn:local_centre}
      \Aut(\mathfrak A_i)\ \le\ Z(\Gamma_i)\qquad\text{for every }i\in\range m,
    \end{equation}
    and
    \begin{equation}\label{eqn:global_centre}
      \Aut(\mathfrak A)\ \le\ Z(\mathbf G(\mathfrak A))\ =\ \prod_{i\in\range m} Z(\Gamma_i).
    \end{equation}
    In particular $\Aut(\mathfrak A)$ is abelian.
  \end{corollary}
  \begin{proof}
    For~\eqref{eqn:local_centre}, apply \cref{lem:ordered_centralises} to the structure $\mathfrak A_i$ and the atomic formula $\Phi(\vec\mu,s,t)$, which has numerical parameters only and defines an ordered generating set for $\Gamma_i$ on $\mathfrak A_i$: every $\pi\in\Aut(\mathfrak A_i)$ commutes with all elements of $\Gamma_i$. As $\pi\in\Aut(\mathfrak A_i)\le\Gamma_i$ by~\eqref{eqn:local_groups}, it lies in $Z(\Gamma_i)$.

    For~\eqref{eqn:global_centre}, every automorphism of $\mathfrak A$ preserves $\preceq^{\mathfrak A}$, hence each colour class setwise, and restricts on $A_i$ to an automorphism of $\mathfrak A_i$, which lies in $\Gamma_i$; so $\Aut(\mathfrak A)\le\mathbf G(\mathfrak A)$. The same formula $\Phi$ defines an ordered generating set for $\mathbf G(\mathfrak A)$ on $\mathfrak A$, so \cref{lem:ordered_centralises} applies with $T = \element$ and every $\pi\in\Aut(\mathfrak A)$ commutes with all elements of $\mathbf G(\mathfrak A)\ge\Aut(\mathfrak A)$. The equality is a well known fact about the centre of a direct product.
  \end{proof}

  Therefore, rather than working with $(\Gamma_i)_{i\in[m]}$, which need not be abelian, we work with their centres $(Z(\Gamma_i))_{i\in[m]}$, which are abelian and still contain $\Aut(\mathfrak A_i)$ --- so that \cref{corol:use_of_ord_op} enables the definition of generating sets for accessible subgroups without using the $\ord$ operator. It remains only to check that the centre is itself presented by a definable ordered generating set.

  \begin{lemma}
    \label[lemma]{lem:enumeration}
    Let $\mathfrak A\in\oPBCG_k$. There is an $\FPC$ formula $\Phi^*(\mu,\vec\nu,s,t)$, with $|\vec\nu| = k$, such that for every colour $i$ the map $\vec\nu\mapsto\Phi^*(\mathfrak A,i,\vec\nu)$ is a bijection from an initial segment of $(\oA)^k$ onto $\{\graph(\gamma)\tq\gamma\in\Gamma_i\}$.
  \end{lemma}
  \begin{proof}
    We carry out a breadth-first search of the Cayley graph of $\Gamma_i$ with respect to the ordered generating set $\Phi(\mathfrak A,i)$, using an inflationary fixed-point.
    Let $X$ be a relation variable of type $\numb\cdot\numb^{k}\cdot\element^2$, the intended value of $X(i,\vec\nu)$ being the graph of the $\vec\nu$-th element of $\Gamma_i$. Stage $0$ sets $X(i,\vec 0) := \graph(\Id)$.
    Given $X$, the composition of the permutation indexed by $\vec\nu$ with the generator indexed by $\vec p$ is defined by
    \[\mathsf{comp}(i,\vec\nu,\vec p,s,t) := \exists u\,\bigl(X(i,\vec\nu,s,u)\wedge\Phi(i,\vec p,u,t)\bigr),\]
    and the pair $(\vec\nu,\vec p)$ is \emph{fresh} if $\mathsf{comp}(i,\vec\nu,\vec p)$ occurs neither in $X$ nor as $\mathsf{comp}(i,\vec\nu',\vec p')$ for a lexicographically smaller pair. It is easily seen that freshness is first-order definable. At each stage, for every fresh $(\vec\nu,\vec p)$, the permutation graph $\mathsf{comp}(i,\vec\nu,\vec p)$ is added to $X(i)$ at the index
    \[\bigl(\#\vec\rho.\ \exists u,v\ X(i,\vec\rho,u,v)\bigr) + \bigl(\#(\vec\nu',\vec p')<_{\mathrm{lex}}(\vec\nu,\vec p).\ (\vec\nu',\vec p')\text{ fresh in }X\bigr).\]
    The operator is inflationary, so its fixed point $\Phi^*$ is $\FPC$-definable. An induction on the number of stages shows that after $\ell$ stages $X(i,\cdot)$ holds exactly the elements of $\Gamma_i$ expressible as a product of at most $\ell$ generators, indexed without repetition. As $\Phi(\mathfrak A,i)$ generates $\Gamma_i$ and $|\Gamma_i|<n^k$, the fixed point is reached after fewer than $n^k$ stages and enumerates all of $\Gamma_i$, the indices used forming an initial segment of $(\oA)^k$.
  \end{proof}

  \begin{corollary}
    \label[corollary]{corol:centre_ordered}
    There is an $\FPC$ formula defining an ordered generating set for $\mathbf Z(\mathfrak A) := \prod_{i\in\range m} Z(\Gamma_i)$ on $\oPBCG_k$.
  \end{corollary}
  \begin{proof}
    By \cref{lem:enumeration} the elements of $\Gamma_i$ are indexed by $\Phi^*$, so
    \[\zeta(\mu,\vec\nu,s,t) := \begin{landcases}
      \Phi^*(\mu,\vec\nu,s,t)\\
      \forall\vec\rho, \Phi^*(\mu,\vec\nu)\circ \Phi^*(\mu,\vec\rho) = \Phi^*(\mu,\vec \rho)\circ \Phi^*(\mu,\vec\nu)\end{landcases}\]
    is first-order in $\Phi^*$ and enumerates $Z(\Gamma_i)$, in the order inherited from $\Phi^*$. Ordering $\bigcup_i Z(\Gamma_i)$ by the pair $(\mu,\vec\nu)$ gives an ordered generating set for the direct product, of size at most $m\cdot n^k$.
  \end{proof}

  To conclude this analysis, the commutativity of the local groups may be assumed without loss of generality.

  \begin{proposition}
    \label[proposition]{prop:wlog_abelian}
    Let $\mathfrak A\in\oPBCG_k$ with local groups $(\Gamma_i)$. Then the family $(Z(\Gamma_i))_{i\in\range m}$ also satisfies~\eqref{eqn:local_groups}, with $|Z(\Gamma_i)| < n^k$. The tower~\eqref{eqn:subgroup_tower_for_colored_graphs} may therefore be instantiated at it, and its base group
    \[\mathbf Z(\mathfrak A)\ :=\ \prod_{i\in\range m} Z(\Gamma_i)\]
    is abelian, contains $\Aut(\mathfrak A)$, and admits an $\FPC$-definable ordered generating set on $\mathfrak A$.
  \end{proposition}
  \begin{proof}
    By~\eqref{eqn:local_centre}, $\Aut(\mathfrak A_i)\le Z(\Gamma_i)$, and $Z(\Gamma_i)\le\Gamma_i\le\Sym(A_i)$, so~\eqref{eqn:local_groups} holds for this family; moreover $|Z(\Gamma_i)|\le|\Gamma_i| < n^k$.
    A direct product of abelian groups is abelian, and $\Aut(\mathfrak A)\le\mathbf Z(\mathfrak A)$ is shown in \cref{corol:opbcg_abelian}. The last claim is \cref{corol:centre_ordered}, the formula $\zeta$ there being an $\FPC$ formula over the signature of $\mathfrak A$.
  \end{proof}

  \subsection{Definability of the automorphism group}

  \begin{theorem}
    \label[theorem]{thm:pbcg}
    $\FPC$ defines the automorphism group of $\oPBCG_k$ structures.
  \end{theorem}
  \begin{proof}
    By \cref{prop:wlog_abelian} we instantiate the tower~\eqref{eqn:subgroup_tower_for_colored_graphs} at the local groups $Z(\Gamma_i)$, so that its base group is $\mathbf G(\mathfrak A) := \mathbf Z(\mathfrak A) = \prod_i Z(\Gamma_i)$. 
    
    As argued above, for structures in $\oPBCG_k$, this subgroup chain witnesses the accessibility of $\Aut_c(\mathfrak A)$ from $\mathbf G(\mathfrak A)$: its indices are bounded by \cref{lem:pbcg_poly_index}, applicable since $|Z(\Gamma_i)|\le|\Gamma_i| < n^k$, and $\Aut_c(\mathfrak A)\le\mathbf G(\mathfrak A) = \prod_i Z(\Gamma_i)$ by~\cref{corol:opbcg_abelian}. 
    
    As $\mathbf G(\mathfrak A)$ is abelian, \cref{corol:use_of_ord_op} applies in place of \cref{thm:ordered_schreier_sims}, and $\FPC$ suffices to define a generating set of any accessible subgroup of $\mathbf G(\mathfrak A)$. 
    It only remains to supply the two formulae required by \cref{corol:use_of_ord_op} for that chain.

    Since $Z(\Gamma_i)$ is generated by the permutations enumerated by $\zeta(\mathfrak A,i)$, and a family of generating sets for the factors generates their direct product, we may take $\varphi_{\mathbf G}(\mu,\vec\nu,s,t) := \zeta(\mu,\vec\nu,s,t)$, which is $\FPC$-definable by \cref{corol:centre_ordered} and satisfies $\GspanFP{\varphi_\mathbf G}{\mu\vec\nu}{\mathfrak A} = \mathbf G(\mathfrak A)$.
    Finally, we define $\varphi_\in$ as follows:
    \[\varphi_\in(\mu,\nu, R) := \forall x,y, \begin{lorcases}
      x\not\in (A_\mu\cup A_\nu)\\
      y\not\in (A_\mu\cup A_\nu)\\
      \exists x',y',\begin{landcases}
      R(x,x')\land
      R(y,y')\\
      E(x,y)\iff E(x',y').
    \end{landcases}
  \end{lorcases}\]
    If $R^{\mathfrak A}$ is the graph of a permutation $\sigma\in G_\lambda$, and $s_\lambda = \{i,j\}$ where $i \le j$, then $(\mathfrak A, i,j,R^{\mathfrak A})\models \varphi_\in$ iff $\sigma$ respects the edge relation on $\mathfrak A_{s_\lambda}$. This is indeed equivalent to $\sigma\in\Aut_c(\mathfrak A_{s_\lambda})$. Note that $\Aut_c(\mathfrak A_{s_\lambda})$ also requires $\Phi$ to be preserved, but every $\sigma\in\mathbf G(\mathfrak A)$ does so automatically. Indeed $\sigma_{\restriction A_i}\in Z(\Gamma_i)$ for each $i$, so $\sigma_{\restriction A_i}$ commutes with ---and hence fixes the graph of--- each permutation enumerated by $\Phi(\mathfrak A,i)$.
    This concludes the proof as we can now apply \cref{corol:use_of_ord_op} on the chain of subgroups defined by $\varphi_\in$, starting from the group $\mathbf G$ for which an ordered generating set is defined by $\varphi_{\mathbf G}$, yielding a generating set for $\mathbf H(\mathfrak A) = \Aut_c(\mathfrak A)$.
  \end{proof}

This result should be contrasted with Lichter's recent breakthrough~\cite{lichterSeparating2023} that $\FP + \rk < \P$. Indeed, the structures shown by Lichter to be indistinguishable in $\FP + \rk$ (which is strictly more expressive than $\FPC$~\cite{dawarLogicsRankOperators2009}) are coloured in such a way that each colour class has a polynomially bounded, abelian, ordered automorphism group. As such, while $\FP + \rk$ cannot distinguish those structures, \cref{thm:pbcg} applies, and $\FPC$ can define their automorphism groups.

While there is generally a polynomial-time reduction from the Graph Isomorphism problem over $\mathcal C$ to the Graph Automorphism problem over $\mathcal C$, this requires the class of graphs $\mathcal C$ at hand to be \emph{union-closed}. In the current context, this is not the case.
Indeed, proving that $\oPBCG_k$ is union-closed would entail, given two such structures $\mathfrak A,\mathfrak B$, to build, within the logic at hand, an ordered enumeration of the automorphisms of the subgraph induced by $A_i\cup B_i$, for each $i$. It is a corollary of Lichter's result and \cref{thm:pbcg} that this cannot be done in neither $\FPC$ nor $\FP + \rk$.

\Cref{lem:ordered_centralises} also delimits the technique of \cref{sec:subgroup_comp} more generally. \Cref{thm:ordered_schreier_sims} is of use only for ambient groups $\mathbf G$ lying in the centraliser of $\Aut(\mathfrak A)$, and its conclusion can concern $\mathbf H(\mathfrak A) = \Aut(\mathfrak A)$ only when that group is abelian.

\section{Conclusion}

Pursuing the study initiated in~\cite{dahanGroup2025} of the ability of candidate logics for $\P$ to express group-theoretic operations, we have shown that generating sets for accessible subgroups are \emph{not} definable in any logic for $\P$ in the general case (\cref{sec:limits}), but are definable in $\FP + \ord$ when the base group admits a definable ordered generating set, and in $\FPC$ when that group is moreover abelian (\cref{sec:subgroup_comp}). This enables the definition of a generating set for the automorphism group of a large class of graphs --- Polynomially Bounded Colour-class graphs with ordered colour-class generators, which include all graphs with abelian colours --- within $\FPC$ (\cref{sec:col_graph_aut}).

It should be noted that \cref{lem:ordered_centralises} implies that this generalization is still constrained to structures whose automorphism groups are abelian. This constraint can be traced back to the assumption that the generating sets are ordered canonically.
This prompts the question whether one can relax the ordering assumption, so that \cref{thm:ordered_schreier_sims} remains true, without restricting the structures to those with $\Aut(\mathfrak A)$ abelian.

There are several other directions in which to extend this study. First, the notion of PBCG is quite artificial, and it would be interesting to see whether the results of \cref{sec:subgroup_comp} enable the definition of generating sets for automorphism groups of more natural classes of graphs.
Second, the results of \cref{sec:limits} and \cref{sec:subgroup_comp} exhibit a strong distinction between the definability of generating sets for accessible subgroups in the ordered and unordered contexts. The existence of middle grounds between those two settings could yield interesting results. For instance, the definability of accessible subgroups could be studied in the context of \emph{unordered} groups with strong structural properties.
This was achieved in~\cite{dahanGroup2025} for abelian groups by shifting the representation of groups from generating sets to a notion of morphism-definability. Whether this can be extended to other classes of groups remains an open question.

\begin{acks}
I thank the anonymous referees of LICS 2026 for their careful reading and for
numerous corrections and suggestions, which have improved this article
considerably.
\end{acks}

\printbibliography
\clearpage
\appendix

\section{Proofs of \cref{sec:subgroup_comp}}
\label{sec:proofs_sc}

\begin{proof}[Proof of \cref{lem:fpc_sift}]
  We first show that we can build a formula $\siftEval(R,R_\sigma,\vec\mu,\vec s,\vec t)$ that defines the trace of the values of $g$ along the for-loop in the evaluation of $\sift((R),\sigma)$. Let $g_i$ be the value of $g$ when entering the for-loop of~\cref{alg:sift} with $\lambda = i$. Precisely, if $i$ is the integer encoded by $\vec a$, and $\lambda = i$ is reached along the computation, $\siftEval(\mathfrak A,R^{\mathfrak A}, \sigma, \vec a) = \graph(g_i)$.
  \begin{equation}
    \label{eqn:sift_formula}
    \siftEval(R,R_\sigma) := \ifp_{\Theta, \vec\mu,\vec s,\vec t} \begin{lorcases}
    \vec \mu = 0 \land \sigma(\vec s) = \vec t\\
    \exists \vec \nu, \begin{landcases}\varphi_\in(\vec\mu,X)[X(\vec x,\vec y) / \chi(\vec x,\vec y)]\\
      \chi(\vec s,\vec t)
    \end{landcases}
   \end{lorcases}\end{equation}
   where
   \[\chi(\vec x,\vec y) := \exists \vec z, \begin{landcases}
      \Theta(\vec \mu - 1, \vec x, \vec z)\\
      R(\vec \mu -1,\vec \nu, \vec y,\vec z)
   \end{landcases}\]
  and where by $\vec \mu - 1$, we mean the tuple encoding the predecessor of the integer encoded by $\vec\mu$ (which is obviously definable in $\FPC$).
  Recall that $R$ is our given transversal table and $\varphi_\in$ the membership formula from \cref{def:witnessed_accessibility}, that takes as input a numerical tuple, and a relation encoding a permutation.

  Let us explain this definition. We claim that this inflationary fixed-point definition is such that, after $k$ iterations, $\Theta$ evaluates to $\bigcup_{i < k} \{i\}\times \graph(g_i)$.
  By definition of the $\ifp$ operator, $\Theta$ initially evaluates to the empty relation, and as such, after one iteration, it is easy to see that $\Theta$ evaluates to $\{0\}\times \graph(\sigma)$. This is in accordance with the fact that $g_0 = \sigma$.

  Let us now show the correctness of the recursive part of this fix-point definition, starting with the definition of $\chi$.
  Notice that, if $R(\mathfrak A,i -1, j) = \graph(\tau)$ and $\Theta(\mathfrak A, i - 1) = \graph(g_{i -1})$ --- that is, if the transversal table contains the permutation $\tau$ at position $(i-1,j)$ and, by induction hypothesis, the permutation resulting of $i-1$ steps of sifting is indeed contained in $\Theta(\mathfrak A,i-1)$ --- then, $\chi(\mathfrak A,R,\Theta,i,j) = \graph(\tau^{-1}g_{i -1})$. 
  That is, $\chi$ is precisely the composition of the permutation defined by $\Theta(\vec\mu - 1)$ with the inverse of the permutation defined by $R(\vec \mu - 1, \vec \nu)$.

  \Cref{eqn:sift_formula} effectively assigns the permutation defined by $\chi$ to $\Theta(i)$, for the value of $\vec\nu$ which encodes the unique $j$ (if such a $j$ exists) such that $R(i-1,j)$ and $\Theta(i - 1)$ belong to the same coset of $\mathbf G_i(\mathfrak A)$ in $\mathbf G_{i -1}(\mathfrak A)$ (as checked using $\varphi_\in$) which is precisely the operation conducted within the then-clause on line 4 of \cref{alg:sift}.

  Now, if at the $i$-th entry in the for-loop, $g_i$ does not belong to any coset defined by $R(\mathfrak A, i)$, none of the two disjuncts in \cref{eqn:sift_formula} hold, and $\Theta(\vec \nu) = \emptyset$ holds for any number of iterations of the formula within the $\ifp$ operator, and for any $\vec\nu$ encoding an integer $j \ge i$.
  From this stems the definition of $\siftLevel$:
  \[\siftLevel(R,R_\sigma,\vec \mu) := \begin{landcases}
    \exists \vec s,\vec t,\siftEval(R,R_\sigma,\vec\mu,\vec s,\vec t)\\
    \forall \vec s,\vec t,\lnot\siftEval(R,R_\sigma,\vec\mu + 1,\vec s,\vec t)
  \end{landcases}
  \]
  From there, $\siftRes$ is also easily definable, as $\siftRes(\mathfrak A, R,\graph(\sigma))$ should evaluate to $\siftEval(\mathfrak A, R,\graph(\sigma),i)$, where $i$ is the minimum between $n^k$ and the number of iterations after which the $\sift$ procedure halts, i.e., the unique integer in $\siftLevel(\mathfrak A,R,\graph(\sigma))$:
  \[
    \siftRes(R,R_\sigma,\vec s,\vec t) := \exists\vec\mu, \siftLevel(R,R_\sigma,\vec\mu)\land \siftEval(R,R_\sigma,\vec\mu,\vec s,\vec t)\qedhere
  \]
\end{proof}

\begin{proof}[Proof of \cref{lem:ins_psgs}]
  We first depict the expected behaviour of $\ins$.
  Consider $\tau$ the result of the partial sifting of $\sigma$, that is $\tau := \perm(\siftRes(\mathfrak A,R^{\mathfrak A},\graph(\sigma)))$ and $i$ the unique index of $R^{\mathfrak A}$ such that $(\mathfrak A,R^{\mathfrak A},\graph(\sigma),i)\models\siftLevel$.
  \begin{itemize}
    \item If $i = n^k - 1$, then $\tau$ should be added to $S$. To make sure that $S$ always contains less than $n^{2|T|}$ elements, we only add $\sigma$ to $S$ if $\tau\not\in\GspanFP{S(\vec\lambda,\vec s,\vec t)}{\vec\lambda}{\mathfrak A}$ (as argued before, this can be tested using the $\ord$ operator).
    \item Otherwise, $\tau$ is not represented in the partial transversal $R(i + 1)$. Therefore, $\tau$ must be added to $R(\vec i + 1,j)$, where $j$ is the minimal value for which $R(\mathfrak A,i + 1,j) = \emptyset$. 
  \end{itemize}
  We now give the definitions of $\insTable$ and $\insList$, where the index $i$ is encoded by the tuple $\vec \mu$, and $j$ by the tuple $\vec\nu$:
  \begin{align*}
    \insTable(R,\sigma,\vec \mu,\vec\nu,\vec s,\vec t) &:= 
    \begin{lorcases}
      R(\vec \mu,\vec\nu,\vec s,\vec t)\\
      \begin{landcases}
        \siftLevel(R,\sigma,\vec \mu)\land \vec \mu\ne n^k -1\\
        \forall \vec {s'},\vec {t'}, \lnot R(\vec \mu,\vec\nu,\vec {s'},\vec {t'})\\
        \forall \vec\nu'<\vec\nu, \exists \vec{s'},\vec{t'}, R(\vec \mu,\vec\nu',\vec{s'},\vec{t'})\\
        \forall\vec\nu'<\vec\nu, \lnot\varphi_\in(\vec\mu, X)[X(\vec x,\vec y) / R_{\rho^{-1}\tau}(\vec\nu',\vec x,\vec y)]\\
        \tau(\vec s) = \vec t\\
      \end{landcases}\\
  \end{lorcases}\\
  \insList(R,S,\sigma,\vec \lambda,\vec s,\vec t) &:=
  \begin{landcases}
    \forall \vec \lambda'<\vec \lambda, \exists \vec{s'},\vec{t'}, S(\lambda',\vec{s'},\vec{t'})\\
    \begin{lorcases}
      S(\vec\lambda,s,t)\\
      \begin{landcases}
        \forall \vec{s'},\vec{t'}, \lnot S(\lambda,\vec{s'},\vec{t'})\\
        \siftLevel(R,\sigma, n^k -1)\\
        \tau(s) = t\\
        \lnot(\tau\in\langle S\rangle)_{\vec\lambda,\vec s, \vec t}
      \end{landcases}
    \end{lorcases}
  \end{landcases}
  \intertext{where $\tau$ and $\rho^{-1}\tau$ are the permutations defined by}
  R_\tau(\vec s,\vec t) &:= \siftRes(T,\sigma, \vec s, \vec t)\\
  R_{\rho^{-1}\tau}(\vec\nu',\vec s,\vec t) &:= \exists \vec z, \begin{landcases}R_{\tau}(\vec x,\vec z)\\
  R(\vec\mu,\vec\nu',\vec y,\vec z)
  \end{landcases}
\end{align*}

The last clause of $\insList$ ensures that the group $\langle S'\rangle$ is strictly larger than $\langle S\rangle$, therefore by a multiplicative factor. As such, since $|\mathbf H(\mathfrak A)|\le |A^T|!$, only $|A^T|\log(|A^T|)$ insertions into $S$ can effectively take place throughout the construction procedure (which ensures that tuples of arity $2|T|$ are large enough to represent indices of any $S$). 
\end{proof}

\begin{proof}[Proof of \cref{corol:proto_psgs}]
  Suppose that we have already added a subset $X$ of the permutations defined by $\varphi_{\mathbf G}$, and suppose additionally that $I^{\mathfrak A}$ is a numerical second-order relation containing exactly those indices $i$ for which we have already inserted $\perm(\varphi_\mathbf G(\mathfrak A,i))$. In this context, the index of the next permutation to be inserted can be defined in $\FPC$, and the updated value of $I$ too:
  \begin{align*}
    \mathrm{step}_{ind}(R',S',I,\vec p) &:= \forall \vec q<\vec p, I(\vec q)\\
  \intertext{With $\mathrm{step}_{ind}$ defined, we can easily define formulae to update $(R,S)$ accordingly:}
    \mathrm{step}_{table}(R',S',I,\vec\mu,\vec\nu,\vec s,\vec t) &:= \begin{lorcases}
      R'(\vec\mu,\vec\nu,\vec s,\vec t)\\
      \begin{landcases}
        \forall \vec {s'},\vec {t'}, \lnot R'(\vec\mu,\vec\nu,\vec {s'},\vec {t'})\\
        \exists \vec p,\begin{landcases}
          \lnot I(\vec p)\land \mathrm{step}_{ind}(R',S',I,\vec p)\\
          \mathrm{ins}_{\mathrm{table}}^*(R',S',R_\sigma,\vec \mu,\vec \nu,\vec s,\vec t)
        \end{landcases}
      \end{landcases}
    \end{lorcases}\\
    \mathrm{step}_{list}(R',S',I,\vec\lambda,\vec s,\vec t) &:= \begin{lorcases}
      S'(\vec\lambda,\vec s,\vec t)\\
      \begin{landcases}
        \forall \vec s',\vec t', \lnot S'(\vec\lambda,\vec s',\vec t')\\
        \exists \vec p, \begin{landcases}
          \lnot I(\vec p)\land \mathrm{step}_{ind}(R',S',I,\vec p)\\
          \mathrm{ins}_{\mathrm{list}}^*(R',S',R_\sigma,\vec \lambda,\vec s,\vec t)
        \end{landcases}
      \end{landcases}
    \end{lorcases}
  \end{align*}
  where the formula $\mathrm{ins}_{\mathrm{table}}^*$ (resp. $\mathrm{ins}_{\mathrm{list}}^*$) is the formula $\insTable$ (resp. $\insList$) where each occurrence of $R_\sigma(\vec x,\vec y)$ has been substituted by $\varphi_{\mathbf G}(\vec p, \vec x, \vec y)$.
  Those $\mathrm{step}$ formulae merely iterate the insertion formulae defined in the last lemma over all permutations defined by $\varphi_{\mathbf G}$. It is only left to consider the simultaneous fixed-point (for a definition of the simultaneous fixed-point operator and the proof that it can be defined with $\ifp$, we refer the reader to~\cite[Theorem 8.2.2]{ebbinghausFinite1995}) of those three formulae, however, we must ensure that the initial values of the fixed-point second-order variables are correctly set:
  \begin{align*}
    \protoTable(R,S,\vec \mu,\vec\nu,\vec s,\vec t) &:= (\sifp_{\vec\mu\vec\nu\vec s\vec t R';\vec\lambda\vec s\vec t S';\vec p I}\mathrm{step}'_{table};\mathrm{step}'_{list};\mathrm{step}_{ind})\\
    \protoList(R,S,\vec \lambda,\vec s,\vec t) &:= (\sifp_{\vec\lambda\vec s\vec t S';\vec\mu\vec\nu\vec s\vec t R';\vec p I}\mathrm{step}'_{list};\mathrm{step}'_{table};\mathrm{step}_{ind})
  \end{align*}
  Indeed, by definition of the $\sifp$ operator, all second-order variables are initiated to the empty relations, while we expect the iteration of the $\mathrm{step}$-formulae to apply initially on the relations $R,S$ given as input to the formulae $\protoTable,\protoList$. This motivates the following definition of the formulae $\mathrm{step'}$:
  \begin{align*}
  \mathrm{step}'_{table}(R,S,R',S',I,\vec\mu,\vec\nu,\vec s,\vec t) &:= \begin{lorcases}
    \forall \vec p, \lnot I(\vec p)\land R(\vec\mu,\vec\nu,\vec s,\vec t)\\
    \mathrm{step}'_{table}(R',S',I,\vec\mu,\vec\nu,\vec s,\vec t)
  \end{lorcases}\\
  \mathrm{step}'_{list}(R,S,R',S',I,\vec\mu,\vec\nu,\vec s,\vec t) &:= \begin{lorcases}
    \forall \vec p, \lnot I(\vec p)\land S(\vec\mu,\vec\nu,\vec s,\vec t)\\
    \mathrm{step}'_{list}(R',S',I,\vec\mu,\vec\nu,\vec s,\vec t)
  \end{lorcases}
  \end{align*}
\end{proof}

\begin{proof}[Proof of \cref{lem:sat_sgs}]
  Following \cref{alg:sgs_cons}, we consider all products of the form $\sigma\tau$ with $\sigma,\tau\in\bigcup \mathcal R\cup\mathcal S$, and insert them (as defined by the $\mathsf{ins}$ formulae) into the PSGS prototype defined by $(R,S)$, until saturation (and thus a fixed-point) is reached.

  It is vital not to insert simultaneously two such products, as inserting two different permutations at the same index would break the very structure of PSGS prototypes\footnote{This would lead inserting the \emph{union} of the graphs of the two permutations, which is not the graph of a permutation}. This is where the ordering on the indexings of $R$ and $S$ show most important, as they enable the definition of an ordered enumeration of products of ordered pairs of elements appearing in $(R,S)$, and we first exhibit this enumeration. In the remainder of this proof, we use tuples of the form $\vec n = (b_n,\vec\mu_n,\vec\nu_n,\vec\lambda_n)$ to represent the variables indexing the following enumeration:
  \newcommand
    \enum
    {\ensuremath{{\mathrm{enum}}}}
  
  \[\enum(R,S,\vec m,\vec n,\vec s,\vec t) :=
  \exists \vec u,\begin{landcases}
    \begin{lorcases}
      b_m = 0\land R(\vec\mu_m,\vec\nu_m,\vec s,\vec u)\\
      b_m = 1\land S(\vec\lambda_m,\vec s,\vec u)
    \end{lorcases}\vspace{0.5em}\\
    \begin{lorcases}
      b_n = 0 \land R(\vec\mu_n,\vec\nu_n,\vec u,\vec t)\\
      b_n = 1\land S(\vec\lambda_n,\vec u,\vec t)
    \end{lorcases}
  \end{landcases}
  \]
  Using the same method as in \cref{corol:proto_psgs}, we can use this enumeration to keep track of all pairs of elements whose product has already been inserted in $(R,S)$ in a second-order relation $I$. After one step, here is how $I$ should be updated:
  {\allowdisplaybreaks
  \begin{align*}
    \mathsf{step}^{\mathsf{sat}}_{ind}(R,S,I,\vec m,\vec n) &:= 
      \begin{landcases}
        \Bij(X)[X(\vec x,\vec y) / \enum(R,S,\vec m,\vec n,\vec x, \vec y)]\\
        \forall \vec m'\vec n'<\vec m\vec n, \begin{lorcases}
          I(\vec m',\vec n')\\
          \lnot\Bij(\enum(R,S,\vec m',\vec n'))
      \end{lorcases}
    \end{landcases}
    \intertext{
      where $\Bij(X)$ is a formula which holds iff $X$ is a $2k$-ary relation which is the graph of a bijection. With this upkeeping of the set of indices defined, an iteration of the saturation process can be defined as follows:
    }
    \mathsf{step}^{\mathsf{sat}}_{table}(R,S,I,\vec\mu,\vec\nu,\vec s,\vec t) &:=
    \begin{lorcases}
      R(\vec\mu,\vec\nu,\vec s,\vec t)\\
      \begin{landcases}
        \forall \vec s',\vec t', \lnot R(\vec\mu,\vec\nu,\vec s',\vec t')\\
        \exists \vec m,\vec n, \begin{landcases}
        \lnot I(\vec m,\vec n)\land \mathsf{step}^{\mathsf{sat}}_{ind}(R,S,I,\vec m,\vec n)\\
        \mathrm{ins}^*_{\mathrm{table}}(R,R_\sigma,\vec \mu,\vec \nu,\vec s,\vec t)
      \end{landcases}
    \end{landcases}
  \end{lorcases}\\
  \mathsf{step}^{\mathsf{sat}}_{list}(R,S,I,\vec\lambda,\vec s,\vec t)&:=
  \begin{lorcases}
    S(\vec\lambda,\vec s,\vec t)\\
    \begin{landcases}
      \forall \vec s',\vec t', \lnot S(\vec\lambda,\vec s',\vec t')\\
      \exists \vec m,\vec n, \begin{landcases}
      \lnot I(\vec m,\vec n)\land \mathsf{step}^{\mathsf{sat}}_{ind}(R,S,I,\vec m,\vec n)\\
      \mathrm{ins}^*_{\mathrm{list}}(R,S,R_\sigma,\vec \mu,\vec \nu,\vec s,\vec t)
    \end{landcases}
  \end{landcases}
\end{lorcases}
\end{align*}}
where, here, $\mathrm{ins}^*_{\mathrm{table}}$ (resp. $\mathrm{ins}^*_{\mathrm{list}}$) is the formula $\insTable$ (resp. $\insList$), where each occurrence of $R_\sigma(\vec x, \vec y)$ is substituted by $\enum(R,S,\vec m,\vec n)$.
Notice that only one permutation is added per iteration of 
\[(R,S,I)\mapsto (\mathsf{step}_{table}(\mathfrak A,R,S,I),\mathsf{step}_{list}(R,S,I), \mathsf{step}_{ind}(R,S,I))\] 

To obtain a completely saturated pair $(R,S)$, it is only left to consider the least fix-point of this iteration. This is once again an application of the simultaneous fixed-point operator, where special care is taken with regard to initial values:
\begin{align*}
  \satTable(R,S) &:= (\sifp_{\vec\mu\vec\nu\vec s\vec t R;\vec \lambda\vec s\vec t S;\vec m,\vec n, I}\mathsf{step}^\mathsf{sat}_{table};\mathsf{step}^\mathsf{sat}_{list}; \mathsf{step}^\mathsf{sat}_{ind})\\
  \satList(R,S) &:= (\sifp_{\vec \lambda\vec s\vec t S;\vec\mu\vec\nu\vec s\vec t R;\vec m,\vec n, I}\mathsf{step}^\mathsf{sat}_{list};\mathsf{step}^\mathsf{sat}_{table}; \mathsf{step}^\mathsf{sat}_{ind})
\end{align*}
Once again, in those expressions, the formulae $\mathsf{step}^\mathsf{sat}_{table}$ and $\mathsf{step}^\mathsf{sat}_{list}$ should be substituted by formulae taking into account the case where $I = \emptyset$ (to detect the initial situation), where they should evaluate respectively to the values of $R$ and $S$ given as input to $\satTable$ and $\satList$. This is done exactly as in the proof of \cref{corol:proto_psgs}. Note that, to avoid the variable capture of $R$ and $S$, those variables should be substituted in the definition of the $\mathsf{step}^{\mathsf{sat}}$ formulae, so that the $\sifp$ operator does not bind them.

It is only left to show that, if $(R,S)$ is a PSGS prototype,
\[(R',S') := (\satTable(\mathfrak A,R,S),\satList(A,R,S))\]
constitutes a PSGS for the same group. That is, for any $i < n^k$, we must show that $R'(i)$ defines a transversal of the cosets of $\mathbf G_{i + 1}(\mathfrak A)$ in $\mathbf G_i(\mathfrak A)$.

The following argument is a direct adaptation of~\cite[Theorem 1]{furstPolynomialtimeAlgorithmsPermutation1980}.
Notice that $(R',S')$ is a PSGS if any element $g\in\mathbf K(\mathfrak A)$ can be written as a product $r_1 r_2\dots r_{n^k -1} s$ with $r_i\in R'(i)$ and $s\in\langle S'\rangle$. (The construction below in fact produces an element $s$ belonging to $S'$ itself, but only $s\in\langle S'\rangle$ is needed, and it is this weaker form that the abelian case of \cref{corol:use_of_ord_op} will use.)
Indeed, in such a case, if $g\in G_i$ and $g\not\in G_{i + 1}$, it must hold that $r_1 = r_2 = \dots = r_{i-1} = \Id$ (as those are the initially set unique representatives in $R'$ of the coset $\Id G_{\lambda + 1}$ in $G_\lambda$), and thus, $r_i$ is a representative of the coset $gG_{i + 1}$.

Now, consider $g\in \mathbf K(\mathfrak A)$, and let us show that $g$ can be written as such a product. Because $\mathbf K(\mathfrak A) = \langle \bigcup R'\cup S'\rangle$, let $s_1\dots s_m$ be a decomposition of $g$ in elements of $\bigcup R'\cup S'$.
Define the \emph{index} of such an element $x\in\bigcup R'\cup S'$, denoted $i(x)$, as the unique $i$ such that $x\in R'(i)$ if such an $i$ exists, or $i(x) = n^k$ if $x\in S'$.
Let $y$ be an element of minimal index whose predecessor $x$ in the sequence $s_1,\dots,s_m$ is such that $i(x)\ge i(y)$. Because $x$ and $y$ both belong to $\bigcup R'\cup S'$, their product has been inserted in the table, and therefore their product admit a decomposition of the form $r_1r_2\dots r_{n^k-1}s$ with $r_i\in R'(i)$ and $s\in S'$. Moreover, since $xy\in G_{i(y)}$, for all $j<i(y)$, $r_j = \Id$. Therefore, rewriting $xy$ as $r_{i(y) + 1}\dots r_{n^k-1}$ in $s_1\dots, s_m$ yields another decomposition of $g$. This process can be iterated until no such $y$ can be found, at which point the sequence is such that, for any $\mu < \nu$, $i(s_\mu) < i(s_\nu)$, i.e. we have found an adequate decomposition of $g$.
\end{proof}

\begin{proof}[Proof of \cref{corol:use_of_ord_op}]
Commutativity has a radical implication on the saturation procedure: it is no longer necessary to consider products $\sigma\tau$ in which either factor belongs to $S$. We establish this first, then observe that it removes the only use made of the $\ord$ operator.

\medskip\noindent
Let $(R,S)$ be a PSGS prototype for the chain $\mathbf G_0(\mathfrak A)\ge\dots\ge\mathbf G_{n^k}(\mathfrak A)$, all of whose groups are abelian, and let $(R',S')$ be obtained by saturating $(R,S)$ with respect to the products $xy$ of pairs $x,y$ lying in a \emph{common row} $R(i)$ only. We claim that $(R',S')$ is already a PSGS.

By the criterion recalled in the proof of \cref{lem:sat_sgs}, it suffices to show that every $g\in\mathbf K(\mathfrak A) = \langle\bigcup R' \cup S'\rangle$ admits a decomposition $r_1r_2\dots r_{n^k - 1}s$ with $r_i\in R'(i)$ and $s\in\langle S'\rangle$.
Let $s_1\dots s_m$ be a decomposition of $g$ into elements of $\bigcup R'\cup S'$, and let $i(x)$ denote the index of $x$, as in the proof of \cref{lem:sat_sgs}. All these elements lie in the abelian group $\mathbf G_0(\mathfrak A)$, so the product is unchanged by permuting its factors; we may therefore assume the sequence sorted, that is $i(s_\lambda)\le i(s_\mu)$ whenever $\lambda\le\mu$. Factors of equal index are then consecutive.

We show by induction on $i<n^k$ that such a sorted decomposition may be brought to one having at most one factor of index $i$. The factors of index $n^k$, which lie in $S'$, are left untouched: their product lies in $\langle S'\rangle$, which is all the criterion requires.
Assume the claim established for all indices below $i$, and let $x_1,\dots,x_p$ be the consecutive factors of index $i$, with $p\ge2$. Both $x_1$ and $x_2$ lie in $R'(i)$, so by construction the product $x_1x_2$ has been sifted, and admits a decomposition into factors of $\bigcup R'\cup S'$. Since $x_1x_2\in\mathbf G_i(\mathfrak A)$, that decomposition has $r_j = \Id$ for every $j<i$, so it consists of factors of index at least $i$, exactly one of which lies in $R'(i)$. Substituting it for $x_1x_2$ and sorting afresh --- which is free, by commutativity --- yields a sorted decomposition of $g$ with $p - 1$ factors of index $i$, and with the factors of index below $i$ untouched, since none were created. Iterating $p-1$ times establishes the claim for $i$.
As no step of the induction increases the number of factors of index below the current one, the process terminates, and delivers a decomposition of the required shape. Hence $(R',S')$ is a PSGS, and no product $rs$ or $sr$ with $r\in\mathcal R$ and $s\in\mathcal S$ ever needs to be sifted.

\medskip\noindent
Every insertion into $S$ now arises either from the given generating set for $\mathbf G(\mathfrak A)$, of size at most $|A|^{|\vec p|}$, or from the sifting of a product of two entries of a common row of $R$. The transversal table has at most $|A|^k$ rows and, by $k$-adequacy, at most $|A|^k$ entries in each, so
\[|R|\le |A|^{2k}\]
holds at every stage of the iteration --- and in particular at its fixed point --- irrespective of how many insertions have taken place. The number of pairs available for sifting is therefore at most $|A|^{4k}$, and
\[|S|\ \le\ |A|^{|\vec p|} + |A|^{4k}.\]
This bound is independent of which permutations are actually inserted, and so requires no test to enforce. This is precisely where the abelian case departs from the general one, where the corresponding bound was obtained (in \cref{lem:ins_psgs}) by using the $\ord$ operator to verify that each insertion strictly increases the order of $\langle S\rangle$.

\medskip\noindent
Accordingly, the indexing tuples of the saturation procedure may be taken of the form $\vec m = (\vec\mu_m,\vec\nu_m)$, tracking positions within $R$ alone, with the enumeration formula adapted to the shortened tuples. In the relation variable $S$ the tuple $\vec\lambda$ may be shortened from $2|T|$ to $4k$, and the last clause of $\insList$ --- the sole occurrence of the $\ord$ operator in the proof of \cref{thm:ordered_schreier_sims} --- deleted. Every remaining formula is unchanged but for the length of $\vec\lambda$, and the argument of \cref{thm:ordered_schreier_sims} applies verbatim. The formula $\varphi_{\mathbf H}$ so obtained lies in $\FPC[\Sigma]$ and defines an ordered generating set for $\mathbf H$, as required.
\end{proof}

\begin{proof}[Proof of~\cref{rem:pcbgkk}]
    We first define a formula that holds on $i,\vec x$ if $\Phi(i,\vec x)$ is the graph of a permutation distinct from all permutations defined by $\Phi(i,\vec y)$, for $\vec y < \vec x$:
    \[\mathsf{new}(\mu,\vec p) := \forall \vec q < \vec p, \exists  s, t, \Phi(\mu,\vec p, s, t)\land\lnot\Phi(\mu,\vec q,\vec s,\vec t)\]
    We are now able to define a reindexing $\Phi'$ of $\Phi$, which is indexed over tuples of $k_1 + 1$ numerical variables instead of $k_2 + 1$, that contains each permutation at most once:
    \[\Phi'(\mu,\vec x, s, t) := \exists \vec p,\begin{landcases}
    \mathsf{new}(\mu,\vec p)\\
    (\# \vec q, (\vec q < \vec p \land \mathsf{new}(\mu,\vec q))) = \vec x\\
    \Phi(\mu,\vec p, s, t)
    \end{landcases}\]
    Because $|\GspanFP{\Phi}{\vec p}{\mathfrak A,i}| < n^{k_1}$ for all $i$, and the permutations enumerated by $\Phi'$ are pairwise distinct, all permutations defined by $\Phi$ must occur in $\Phi'$. (Note that if $\Phi$ contained $ v < n^{k_1}$ permutations, $\Phi'(\mathfrak A, v + c)$ is empty for all $c$.)
  \end{proof}

\end{document}